\documentclass[a4paper, UKenglish, cleveref, autoref, thm-restate, numberwithinsect]{lipics-v2021}

\pdfoutput=1 
\hideLIPIcs  

\title{The Algebras for Automatic Relations\thanks{Licensed under \href{http://creativecommons.org/licenses/by/4.0/}{CC BY 4.0}.}}
\titlerunning{The Algebras for Automatic Relations}

\author{Rémi Morvan}{LaBRI (Univ. Bordeaux, CNRS \& Bordeaux INP), France \and \url{https://www.morvan.xyz} }{remi.morvan@u-bordeaux.fr}{https://orcid.org/0000-0002-1418-3405}{}

\authorrunning{R. Morvan} 

\Copyright{Rémi Morvan} 



\category{} 
\relatedversion{}
\acknowledgements{We thank Pablo Barceló, Mikołaj Bojańczyk, and Diego Figueira for
helpful discussions, and some anonymous reviewers for valuable feedback.}
\nolinenumbers 

\EventEditors{John Q. Open and Joan R. Access}
\EventNoEds{2}
\EventLongTitle{42nd Conference on Very Important Topics (CVIT 2016)}
\EventShortTitle{CVIT 2016}
\EventAcronym{CVIT}
\EventYear{2016}
\EventDate{December 24--27, 2016}
\EventLocation{Little Whinging, United Kingdom}
\EventLogo{}
\SeriesVolume{42}
\ArticleNo{23}

\usepackage{adforn}
\usepackage{stmaryrd}
\usepackage{multicol}
\usepackage{dutchcal} 
\renewcommand\S{\adfS}
\renewcommand\phi{\varphi}
\usepackage{thm-restate}
\usepackage[export]{adjustbox} 
\usepackage{stackengine}

\usepackage{tikz}
\usetikzlibrary{cd}

\usetikzlibrary{calc, positioning}
\usetikzlibrary{automata}
\tikzset{
    initial text={},
    accepting/.style=accepting by double
}

\tikzcdset{
	arrow style=tikz,
	diagrams={>={Straight Barb[scale=0.8]}}
}
\usepackage{xcolor}

\definecolor{Dark Ruby Red}{HTML}{6a211f}
\definecolor{Dark Blue Sapphire}{HTML}{164a59}
\definecolor{Dark Gamboge}{HTML}{be7c00}

\definecolor{Desire}{HTML}{eb3b5a} 
\definecolor{Boyzone}{HTML}{2d98da} 
\definecolor{Royal Blue}{HTML}{3867d6} 
\definecolor{NYC Taxi}{HTML}{f7b731} 
\definecolor{Algal Fuel}{HTML}{20bf6b} 
\definecolor{Innuendo}{HTML}{a5b1c2} 
\definecolor{Twinkle Blue}{HTML}{d1d8e0} 
\definecolor{Gloomy Purple}{HTML}{8854d0} 

\colorlet{cBlue}{Royal Blue}
\colorlet{cYellow}{NYC Taxi}
\colorlet{cGreen}{Algal Fuel}
\colorlet{cRed}{Desire}
\colorlet{cGrey}{Innuendo}
\colorlet{cLightGrey}{Twinkle Blue}
\colorlet{cPurple}{Gloomy Purple}
\usepackage[xcolor, hyperref, cleveref, notion, quotation, electronic]{knowledge}
\knowledgeconfigure{quotation, protect quotation={tikzcd}}
\knowledgeconfigure{diagnose line=true, diagnose bar=true}

\IfKnowledgePaperModeTF{
}{
    \knowledgestyle{intro notion}{color={Dark Ruby Red}, emphasize}
    \knowledgestyle{notion}{color={Dark Blue Sapphire}}
    \AtBeginDocument{
		\hypersetup{
			colorlinks=true,
			breaklinks=true,
			linkcolor={Dark Blue Sapphire}, 
			citecolor={Dark Blue Sapphire}, 
			filecolor={Dark Blue Sapphire}, 
			urlcolor={Dark Blue Sapphire},
		}
	}
    \IfKnowledgeElectronicModeTF{
    }{
        \knowledgeconfigure{anchor point color={Dark Ruby Red}, anchor point shape=corner}
        \knowledgestyle{intro unknown}{color={Dark Gamboge}, emphasize}
        \knowledgestyle{intro unknown cont}{color={Dark Gamboge}, emphasize}
        \knowledgestyle{kl unknown}{color={Dark Gamboge}}
        \knowledgestyle{kl unknown cont}{color={Dark Gamboge}}
    }
}

\ExplSyntaxOn
\RenewDocumentCommand\withkl{mm}{
\int_gincr:N\knowledge_inner_modifier_count_int
\cs_gset:cpx
{\knowledge_inner_command:}
{\exp_not:N\cs_gset:Npn
\exp_not:c{\knowledge_inner_command:}
{\knowledge_inner_modifer_re_tl\knowledge_kl_modifiers_tl\exp_not:n{#1}}
\knowledge_kl_modifiers_tl\exp_not:n{#1}}
\knowledge_kl_modifiers_reset:
#2
\int_gdecr:N\knowledge_inner_modifier_count_int
}
\ExplSyntaxOff

\newrobustcmd\defeq{\mathrel{\hat{=}}}
\newrobustcmd\dom{\mathop{\textrm{dom}}}
\newrobustcmd\?[1]{\mathbf{#1}}
\newrobustcmd\+[1]{\mathcal{#1}}
\newrobustcmd\B[1]{\mathbb{#1}}
\newrobustcmd\pset{\mathop{\mathcal{P}\mkern-1mu}}
\newrobustcmd\cat[1]{\mathsf{#1}}
\newrobustcmd\surj{\twoheadrightarrow}
\newrobustcmd\lBrack{\llbracket}
\newrobustcmd\rBrack{\rrbracket}

\newrobustcmd\N{\mathbb{N}}
\newrobustcmd\Np{\mathbb{N}_{>0}}
\newrobustcmd\Z{\mathbb{Z}}
\newrobustcmd\id{\textrm{id}}

\knowledgenewrobustcmd\Set[1][\+S]{\cmdkl{\cat{Set}^{\smash{#1}}}}
\knowledgenewrobustcmd\Pos[1][\+S]{\cmdkl{\cat{Pos}^{\smash{#1}}}}
\knowledgenewrobustcmd\Dep[1][\+S]{\cmdkl{\cat{Dep}^{\smash{#1}}}}

\newcommand{\proofcase}[1]{\adforn{39}~\emph{#1}~}
\renewrobustcmd{\emptyset}{\varnothing}

\knowledgenewrobustcmd\equivclass[2]{\cmdkl{[}#1\cmdkl{]^{\smash{#2}}}}

\newrobustcmd\decisionproblem[3]{%
\AP\begin{center}%
\fbox{\begin{tabular}{rl}%
{\textit{Problem}}: & #1 \\%
{\textit{Input}}: & #2 \\%
{\textit{Question}}: & #3%
\end{tabular}}%
\end{center}%
}

\usepackage{xfrac}
\renewrobustcmd\ll{\kl[\ll]{\text{\sfrac{\footnotesize\textsc{l}\,}{\,\footnotesize\textsc{l}}}}\hspace{.1em}}
\knowledge{\ll}{notion}
\knowledgenewrobustcmd\lb{\cmdkl{\text{\sfrac{\footnotesize\textsc{l}\,}{\,\footnotesize\textsc{b}}}}\hspace{.1em}}
\knowledgenewrobustcmd\bl{\cmdkl{\text{\sfrac{\footnotesize\textsc{b}\,}{\,\footnotesize\textsc{l}}}}\hspace{.1em}}

\definecolor{Jalapeno}{HTML}{b71540}
\newrobustcmd{\todo}[1]{{\color{Jalapeno}{Todo: #1}}}

\newrobustcmd{\pad}{%
	\hspace{.05em}\rule{.3em}{.075em}\hspace{.05em}
}
\knowledgenewrobustcmd\SigmaPair[1][\Sigma]{\cmdkl{#1^{\smash{2}}_{\smash{\hspace{.03em}\raisebox{.18em}{\scriptsize\pad}}}}\!}
\newrobustcmd\pair[2]{\left(\begin{smallmatrix}%
	#1\\%
	#2%
\end{smallmatrix}\right)}%
\knowledgenewrobustcmd\WellFormed[1][\Sigma]{\cmdkl{\textsf{WellFormed}_{#1}}}
\knowledgenewrobustcmd\proj[1]{\cmdkl{\underline{#1}}} 
\knowledgenewrobustcmd\projP[1]{\cmdkl{\underline{#1}}}

\knowledgenewrobustcmd\typeW[2]{#1_{\cmdkl{#2}}}
\knowledgenewrobustcmd\types{\cmdkl{\mathcal{T}}\!}
\knowledgenewrobustcmd\typesEps{\cmdkl{\mathcal{T}_{1}}\!}
\knowledgenewrobustcmd\concattype{\cmdkl{\cdot}}
\knowledgenewrobustcmd\type[3][]{#2^{#1}_{\cmdkl{#3}}}

\knowledgenewrobustcmd\disunion{%
	\mathop{\cmdkl{\text{\raisebox{-2pt}{\LARGE$\uplus$}}}}%
}

\knowledgenewrobustcmd\dep[1][]{%
	\mathrel{\cmdkl{\asymp#1}}%
}

\knowledgenewrobustcmd\negrel{\cmdkl{\neg}}

\knowledgenewrobustcmd\Sync[2][2]{\cmdkl{\?S_{#1}#2}}
\knowledgenewrobustcmd\SyncP[2][2]{\cmdkl{\?S^+_{#1}#2}}

\knowledgenewrobustcmd\consol[1]{\cmdkl{#1^{\smash{0}}}} 

\knowledgenewrobustcmd\projtype{\cmdkl{\pi_{\smash{\textsf{type}}}}}
\knowledgenewrobustcmd\projval{\cmdkl{\pi_{\smash{\textsf{val}}}}}
\knowledgenewrobustcmd\typemap{\cmdkl{\mathrm{type}}}

\newrobustcmd\Acc{\textrm{Acc}}
\newrobustcmd\Bcc{\textrm{Bcc}}
\knowledgenewrobustcmd\inducedmor[1]{\cmdkl{\smash{\tilde{#1}}}}
\knowledgenewrobustcmd\inducedalg[1]{\cmdkl{\?A_{#1}}}
\knowledgenewrobustcmd\inducedmorSg[1]{\cmdkl{\smash{\tilde{#1}}}}
\knowledgenewrobustcmd\inducedalgSg[1]{\cmdkl{\?A_{#1}}}

\knowledgenewrobustcmd\SyntSA[1]{\cmdkl{\?A_{\smash{#1}}}}
\knowledgenewrobustcmd\SyntSAM[2][]{\cmdkl{\eta_{\smash{#2}}^{#1}\!}}
\knowledgenewrobustcmd\SyntSAP[1]{\cmdkl{\?A_{\smash{#1}}}}
\knowledgenewrobustcmd\SyntSAMP[2][]{\cmdkl{\eta_{\smash{#2}}^{#1}\!}}
\newrobustcmd\consolSyntSAP[1]{\kl[\SyntSAP]{\?A_{#1}^{\kl[\consol]{\,\smash{0}}}}}
\newrobustcmd\consolSyntSAMP[1]{\kl[\SyntSAMP]{\eta_{#1}^{\kl[\consol]{\,\smash{0}}}}}
\knowledgenewrobustcmd\sgpSyntSAMP[1]{\cmdkl{\eta_{#1}^{\smash{\textsf{sgp}}}}}

\knowledgenewrobustcmd\quotient[2]{\cmdkl{{#1}/{#2}}}
\knowledgenewrobustcmd\residual[2][]{#2^{\cmdkl{-1}}_{\kl[\type]{#1}}}

\newcommand{\asympbar}{\mathrel{\raisebox{-.1em}{\scalebox{.6}{\ensuremath{\setstackgap{S}{0pt}%
	\stackMath\mathrel{\stackon{\frown}{\stackon{\raisebox{.135em}{\rule{.6em}{.08em}}}{\smile}}}}}}}}

\knowledgenewrobustcmd\congr[1]{%
	\mathrel{\cmdkl{\asympbar_{#1}}}%
}

\knowledgenewrobustcmd\compos{%
	\mathbin{\cmdkl{\circ}}%
}

\newrobustcmd{\coslicesurj}{\!\!\text{\rotatebox[origin=c]{60}{$\relbar\!\joinrel\!\twoheadrightarrow$}}\!}
\newrobustcmd{\coslice}{\!/}

\knowledgenewrobustcmd\projL[1]{\cmdkl{#1^{\textrm{sync}}}}
\knowledgenewrobustcmd\projA[1]{\cmdkl{#1^{\textrm{sync}}}}

\knowledgenewrobustcmd\corrAR{%
	\mathop{\cmdkl{\to}}%
}
\knowledgenewrobustcmd\corrRA{%
	\mathop{\cmdkl{\to}}%
}
\knowledgenewrobustcmd\corr{%
	\mathop{\cmdkl{\leftrightarrow}}%
}

\knowledgenewrobustcmd\Nil{\cmdkl{\+{(\mkern-1mu c\mkern-1mu o\mkern-2mu )\mkern-1mu F\mkern-2mu i\mkern-1mu n}}}
\knowledgenewrobustcmd\NilS{\cmdkl{\B{N}\textrm{il}}}
\knowledgenewrobustcmd\LocTriv{\cmdkl{\+{L\mkern-1mu I}}}
\knowledgenewrobustcmd\LocTrivS{\cmdkl{\B{LI}}}

\knowledgenewrobustcmd\idem{\cmdkl{\omega}}
\knowledgenewrobustcmd\dist{%
	\mathop{\cmdkl{\textbf{d}}}%
}
\knowledgenewrobustcmd\rist{%
	\mathop{\cmdkl{\textbf{r}}}%
}
\knowledgenewrobustcmd\typing{%
	\mathop{\cmdkl{\textbf{t}}}%
}

\knowledgenewrobustcmd\profSy[1]{\cmdkl{\smash{\widehat{#1}}}}
\knowledgenewrobustcmd\profSg[1]{\cmdkl{\smash{\widehat{#1}}}}

\knowledgenewrobustcmd\profeq{%
	\mathrel{\cmdkl{=}}%
}
\knowledgenewrobustcmd\profimp{%
	\mathrel{\cmdkl{\to}}%
}
\knowledgenewrobustcmd\profequiv{%
	\mathrel{\cmdkl{\leftrightarrow}}%
}
\knowledgenewrobustcmd\profdep{%
	\mathrel{\cmdkl{\asymp}}%
}

\knowledgenewrobustcmd\semprofdep[1]{\cmdkl{\lBrack} #1 \cmdkl{\rBrack}}

\knowledgenewrobustcmd\substitWord[2][\typing]{\cmdkl{#2^{#1}}}
\knowledgenewrobustcmd\substitProf[2][\typing]{\cmdkl{#2^{#1}}}
\knowledgenewrobustcmd\inducedProfdep[1]{\cmdkl{#1^{\smash{\textsf{sync}}}}}

\knowledgenewrobustcmd\EqVar[1]{\cmdkl{\+E_{\smash{#1}}}}

\knowledgenewrobustcmd\distvar[1][\B{V}]{%
	\mathop{\cmdkl{\textbf{d}}_{#1}}%
}
\knowledgenewrobustcmd\ristvar[1][\B{V}]{%
	\mathop{\cmdkl{\textbf{r}}_{#1}}%
}
\knowledgenewrobustcmd\simvar[1][\B{V}]{\mathrel{\cmdkl{\asympbar_{#1}}}}
\knowledgenewrobustcmd\freepro[2][\B{V}]{\cmdkl{\smash{\widehat{\mathbf{F}}}_{#1}#2}}
\knowledgenewrobustcmd\freeproproj[1][\B{V}]{\cmdkl{\pi_{#1}}}

\newrobustcmd\Pl[1][\B{V}]{%
	\mathop{%
		\withkl{\kl[\Pl]}{%
			\cmdkl{\mathcal{P\mkern-1mu l}_{\mkern-2mu #1}}
		}%
	}%
}
\knowledge{\Pl}{notion}

\newrobustcmd\Cl[2][\B{V}]{%
	\mathop{%
		\withkl{\kl[\Cl]}{%
			\cmdkl{\mathcal{C\mkern-1mu l}_{\mkern-2mu #1}^{\smash{#2}}}
		}%
	}%
}
\knowledge{\Cl}{notion}

\knowledgenewrobustcmd\MonadSync[1][2]{\cmdkl{\B{S}_{#1}}}
\knowledgenewrobustcmd\MonadSyncP[1][2]{\cmdkl{\B{S}^+_{#1}}}

\newrobustcmd\PictureTypedSet[5]{
	\begin{center}
		\begin{tikzpicture}
			\node[draw, circle] (ll) at (0, 0)  {$\ll$};
			\node[draw, circle] (lb) [above right= .7cm and 1.4cm of ll] {$\lb$};
			\node[draw, circle] (bl) [below right= .7cm and 1.4cm of ll] {$\bl$};
			
			\draw (ll) edge[->, loop left] node[left] {$#1$} (ll);
			\draw (ll) edge[->, dashed] node[midway, fill=white, opacity=.8] {$#2$} (lb);
			\draw (lb) edge[->, loop right] node[right] {$#3$} (lb);
			\draw (ll) edge[->, dashed] node[midway, fill=white, opacity=.8] {$#4$} (bl);
			\draw (bl) edge[->, loop right] node[right] {$#5$} (bl);
		\end{tikzpicture}
	\end{center}
}

\newrobustcmd{\sat}{\mathrel{\kl[\sat]{\vDash}}}

\AtEndPreamble{%
	\newtheorem{fact}[theorem]{Fact}
	\crefname{fact}{Fact}{Facts}
	\theoremstyle{acmdefinition}
	\newtheorem{question}[theorem]{Question}
	\crefname{question}{Question}{Questions}
}



\input{kl/basic.kl}
\input{kl/languages.kl}
\input{kl/relations.kl}
\input{kl/types.kl}
\input{kl/synchronous-algebras.kl}
\input{kl/pseudovarieties.kl}
\input{kl/profinite.kl}
\input{kl/semigroups.kl}
\input{kl/misc.kl}

\begin{document}

\maketitle

\begin{abstract}
	We introduce ``synchronous algebras'', an algebraic structure tailored to recognize automatic relations ("aka" synchronous relations, or regular relations). They are the equivalent of monoids for regular languages, however they conceptually differ in two points: first, they are typed and second, they are equipped with a dependency relation expressing constraints between elements of different types.

	The interest of the proposed definition is that it allows to lift, in an effective way, pseudovarieties of regular  languages to that of synchronous relations, and we show how algebraic characterizations of pseudovarieties of regular languages can be lifted to the pseudovarieties of synchronous relations that they induce.
	Since this construction is effective, this implies that the membership problem is decidable
	for (infinitely) many natural classes of automatic relations.
	A typical example of such a pseudovariety is the class of ``group relations'', defined as the relations recognized by finite-state synchronous permutation automata.

	In order to prove this result, we adapt two pillars of algebraic language theory to synchronous algebras: (a) any relation admits a syntactic synchronous algebra recognizing it, and moreover, the relation is synchronous if, and only if, its syntactic algebra is finite and (b) classes of synchronous relations with desirable closure properties ("ie" pseudovarieties) correspond to pseudovarieties of synchronous algebras.
\end{abstract}

\bigskip
\noindent\adfsmallleafright~
This pdf contains internal links: clicking on a "notion@@notice" leads
to its \AP ""definition@@notice"".

\section{Introduction}
\label{sec:intro}

\subsection{Background}


The landscape of rationality for $k$-ary relations of finite words ($k \geq 2$) is far more complex than for languages---recall that languages can be seen as unary relations of finite words---as depicted in \Cref{fig:landscape-rationality} on page~\pageref{fig:landscape-rationality}. Perhaps the most natural class is that of \AP""rational relations"", defined as relations accepted by non-deterministic two-tape automata---an input $(u,v)$ is described by writing $u$ on the first tape and $v$ and the second tape---that can move its two heads independently, from left to right---see \cite[\S 2.1]{carton_decision_2006} for a formal definition. For instance, the suffix relation is "rational".

Our paper focuses on "synchronous relations", "aka" "automatic relations" or "regular relations", defined as the "rational relations" that can be recognized by
"synchronous automata",
a subclass of the machines described above obtained by keeping a single head that
moves synchronously from left to right, reading one pair of letters after the other; we add
padding symbols $\pad$ at the end of the shorter word---see
\Cref{fig:ex-sync-auto}.
While the suffix relation is not "synchronous", typical examples include the prefix relation,
the same-length relation, etc.
"Synchronous relations" play a central role in the
definitions of automatic structures---introduced by Hodgson \cite{hodgson_theories_1976,hodgson_direct_1982,hodgson_decidabilite_1983} and rediscovered by Khoussainov \& Nerode \cite{goos_automatic_1995}, see \cite[\S XI, pp.~627--762]{blumensath_monadic_2023}.
They also have been studied in the context of graph databases \cite[Definition 3.1, p.7 \& Theorem 6.3, p.~13]{barcelo_expressive_2012}, see \cite[\S 8, p.~17]{figueira_foundations_2021} for more context \& results on \emph{extended} conjunctive regular path queries.

\begin{figure}[htb]
	\begin{center}
		\includegraphics[width=.4\linewidth,valign=m]{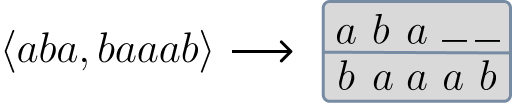}
		\qquad
		\begin{tikzpicture}[shorten >= 1pt, node distance = 1.8cm, on grid, baseline]
			\AP\small
			\node[state, initial left, accepting] (q0) {}; 
			\node[state] (q1) [right =of q0] {};
			\path[->]
				(q0) edge[loop above] node[font=\scriptsize] {$\pair{a}{a}, \pair{b}{b}, \textrm{Pad}\hspace{1.5em}$} (q0) 
				(q0) edge[bend left=20] node[above, font=\scriptsize] {$\pair{a}{b}, \pair{b}{a}$} (q1)
				(q1) edge[bend left=20] node[below, font=\scriptsize] {$\pair{a}{b}, \pair{b}{a}$} (q0)
				(q1) edge[loop above] node[font=\scriptsize] {$\hspace{1.5em}\pair{a}{a}, \pair{b}{b}, \textrm{Pad}$} (q1);
		\end{tikzpicture}
	\end{center}
	\caption{
		\label{fig:ex-sync-auto}
		Encoding a pair of words of $\Sigma^* \times \Sigma^*$ into an element of
		$(\SigmaPair)^*$ where $\SigmaPair \defeq (\Sigma \times \Sigma) \,\cup\,
		(\Sigma \times \{\pad\}) \,\cup\,
		(\{\pad\} \times \Sigma)$ (left) and 
		a deterministic complete "synchronous automaton" (right)
		over $\Sigma = \{a,b\}$ accepting the binary relation
		of pairs $(u,v)$ such that the number of $a$'s in $u_1\hdots u_k$
		and in $v_1\hdots v_k$ are the same mod $2$, where $k = \min(|u|, |v|)$. 
		$\textrm{Pad}$ denotes the set of transitions $\{\pair{a}{\pad}, \pair{b}{\pad}, \pair{\pad}{a}, \pair{\pad}{b}\}$.
	} 
\end{figure}

\begin{remark}
	All our results are described for binary relations, but can be extended to
	$k$-ary synchronous relations, see \Cref{sec:discussion}.
\end{remark}

"Synchronous relations" stand at the frontier between expressiveness and undecidability: for instance, Carton, Choffrut and Grigorieff showed that it is decidable whether an "automatic relation" is ""recognizable"" \cite[Proposition 3.9, p.~265]{carton_decision_2006}, meaning that
it can be written as a finite union of Cartesian products of regular languages.\footnote{For 
instance, the relation ``having the same length modulo $2$'' is "recognizable", since it can be 
written as $(aa)^*\times (aa)^* \cup
a(aa)^*\times a(aa)^*$.}\footnote{The problem was latter shown to be "NL"-complete and
"PSpace"-complete depending on whether the input automaton is deterministic or not
in \cite[Theorem 1, p.~3]{barcelo_monadic_2019}.}
"Synchronous relations" are effectively closed under Boolean operations---see "eg" \cite[Lemma XI.1.3, p.~627]{blumensath_monadic_2023},
and moreover, inclusion (and subsequent problems: universality, emptiness, equivalence…) is decidable for them,
by reduction to classical automata,
contrary to the equivalence problem over "rational relations"
which is undecidable \cite[Theorem 8.4, p.~81]{berstel_transductions_1979}.

However, some seemingly easy problems are undecidable: Köcher showed that
it is undecidable if the (infinite) graph defined by a "synchronous relation" is
2-colourable---\cite[Proposition 6.5, p.~43]{kocher_analyse_2014}, and Barceló, Figueira and Morvan showed that undecidability also
holds for regular 2-colourability \cite[Theorem 4.4, p.~8]{barcelo_separating_2023}.
On the other hand, one can decide
if said graph contains an infinite clique, see \cite[Corollary 5.5, p.~32]{kuske_natural_2010}:
this is a consequence of \cite[Theorem 3.20, p.~185]{Rubin_2008}. 

\subsection{Motivation}

Any "synchronous relation" can be seen as a regular language over the alphabet
$\SigmaPair \defeq (\Sigma \times \Sigma) \,\cup\,
(\Sigma \times \{\pad\}) \,\cup\, (\{\pad\} \times \Sigma)$ of pairs. On the other hand any regular language $L$ over $\SigmaPair$
produces a "synchronous relation" when intersected with the language of all
"well-formed words"---namely words where the padding symbols are consistently placed;
see \Cref{sec:preliminaries} for precise definitions. In fact, the semantics
of "synchronous automata" such as the one in \Cref{fig:ex-sync-auto} is precisely defined this way:
it is the intersection of the ``classical semantic'' of the automaton, seen as an NFA, intersected
with "well-formed words".

\begin{figure}[htbp]
	\begin{center}
		\includegraphics[width=.4\linewidth]{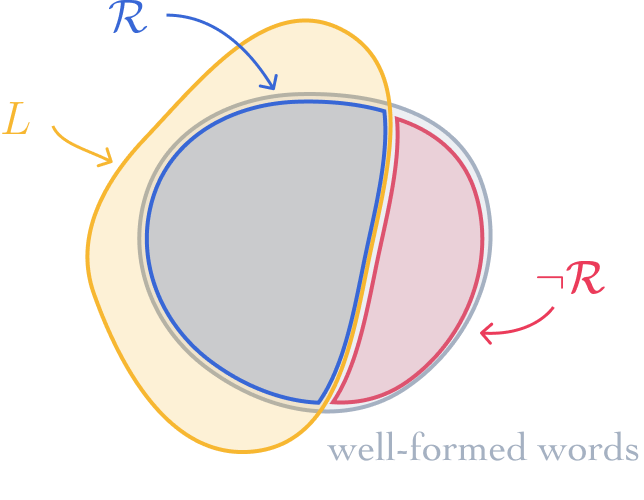}
	\end{center}
	\caption{
		\label{fig:projection}
		Drawing in $(\SigmaPair)^*$ of a "$\+V$-relation" $\+R$ and $\negrel\+R \defeq \{(u,v) \in \Sigma^*\times \Sigma^* \mid (u,v) \not\in \+R\}$, where $\+R$ is defined as $L \cap \WellFormed$ with $L \in \+V$.
	} 
\end{figure}

In particular, a class $\+V$ of regular languages over
$\SigmaPair$
("e.g." first-order definable languages, group languages, etc.) induces a class of so-called
"$\+V$-relations", defined as the relations over $\Sigma$ obtained as the intersection of some 
language of $\+V$ with "well-formed words", see \Cref{fig:projection}.
For instance, the relation of \Cref{fig:ex-sync-auto}
is a "$\+V$-relation" where $\+V$ is the class of all group languages---these relations can be 
alternatively described as those recognized by a deterministic complete synchronous automaton whose 
transitions functions are permutations of states.

\begin{question}
	\label{quest:V-relations}
	Given a class $\+V$ of languages, can we characterize and decide the class of "$\+V$-relations"?
\end{question}

As we will see in \Cref{ex:group-languages}, for a "relation" to be $\+V_{\SigmaPair}$
is not necessary for it to be a "$\+V$-relation".

\subsection{Contributions}

We answer positively to this question.
For this we first need to develop an algebraic theory of "synchronous relations",
which enables us to prove the lifting theorem. In short, the "lifting theorem" states that algebraic characterizations of classes of word languages can be lifted in a canonical way to algebraic characterizations of classes of word relations.

The algebraic approach usually provides more than decidability: it attaches
canonical algebras to languages/relations ("eg" monoids for languages of finite words), and often simple ways to characterize complex properties ("eg" first-order definability, see "eg" \cite[Theorem 2.6, p.~40]{bojanczyk_languages_2020}).
Our "synchronous algebras" differ from monoids in two points:
\begin{itemize}
	\item they are typed---a quite common feature in algebraic language theory, shared "eg" by $\omega$-semigroups \cite[\S 4.1, p.~91]{perrin_infinite_2004};
	\item they are equipped with a "dependency relation", which expresses constraints between 
	elements of different types---to our knowledge, this feature is entirely novel.\footnote{Note that algebras equipped with binary relations have been studied before, "eg" Pin's ordered 
	$\omega$-semigroups---see \cite[\S 2.4, p.~7]{pin_positive_1998}---but the constraints (here the 
	orderings) are always defined between elements of the \emph{same type}.}
\end{itemize}

Importantly, some variations are possible on the definition of "synchronous algebras":
in particular, one could get rid of the notion of "dependency relation" and 
\Cref{lem:syntactic-morphism-theorem,lem:eilenberg-sy} would still hold.
However, we show in \Cref{apdx:counterexample} that these
simplified synchronous algebras cannot characterize the property of being a "$\+V$-relation".
Therefore, the notion of "dependency" seems necessary to tackle \Cref{quest:V-relations}.
Moreover, we show that these algebras arise from a monad, but to our knowledge none of the 
meta-theorems developing algebraic language theories over monads apply to it,
see \Cref{apdx:monads} for more details.

We show that assuming that $\+V$ is a "$*$-pseudovariety of regular languages"---in short, a class of regular languages with desirable closure properties---, then the algebraic characterization of $\+V$ can be easily lifted to characterize "$\+V$-relations".

\begin{restatable*}[\reintro{Lifting theorem: Elementary Formulation}]{theorem}{liftingtheoremmonoids}
	\label{thm:lifting-theorem-monoids}
	Given a "relation" $\+R$ and a "$\ast$-pseudovariety of regular languages" $\+V$
	"corresponding@@EilenbergSg" to a "pseudovariety of monoids" $\B{V}$,
	the following are equivalent:
	\begin{enumerate}
		\item $\+R$ is a "$\+V$-relation",
		\item $\+R$ is "recognized@@sync" by a finite "synchronous algebra" $\mathbf{A}$
			whose "underlying monoids" are all in $\mathbb{V}$,
		\item all "underlying monoids" of the "syntactic synchronous algebras" $\SyntSA{\+R}$ of
			$\+R$ are in $\mathbb{V}$.
	\end{enumerate} 
\end{restatable*}

This theorem rests on a solid algebraic theory. 
First, we show the existence of "syntactic algebras@@sync" (\Cref{lem:syntactic-morphism-theorem}): 
each relation $\+R$ admits a unique canonical and minimal algebra $\SyntSA{\+R}$, which is finite 
"iff" the relation is "synchronous",
and then, we exhibit a correspondence between classes of finite algebras and classes of
synchronous relations (\Cref{lem:eilenberg-sy})---we assume suitable closure properties; these classes are called ``pseudovarieties''.
While the proof structures of \Cref{lem:syntactic-morphism-theorem,lem:eilenberg-sy} follow the classic proofs, see "eg" \cite{pin_mathematical_2022},
the "dependency relation" has to be taken into account quite carefully, leading for instance
to a surprising definition of "residuals", see \Cref{def:residuals}.

\subparagraph*{Organization.} After giving preliminary results in \Cref{sec:preliminaries}, we introduce
the "synchronous algebras" in \Cref{sec:synchronous-algebras} and show the existence of
"syntactic algebras@@sync". We then proceed to prove the "lifting theorem@@monoids" for 
"$*$-pseudovarieties@@reglang" in \Cref{sec:lifting-theorem}, and after introducing "$*$-pseudovarieties of synchronous relations", we provide a more algebraic reformulation of the "lifting 
theorem@@monoidspseudovar" (\Cref{thm:lifting-theorem-monoids-pseudovarieties}).
We conclude the paper with
a short discussion in \Cref{sec:discussion}.

\subsection{Related Work}

\begin{figure}[htb]
	\centering
	\includegraphics[width=\linewidth]{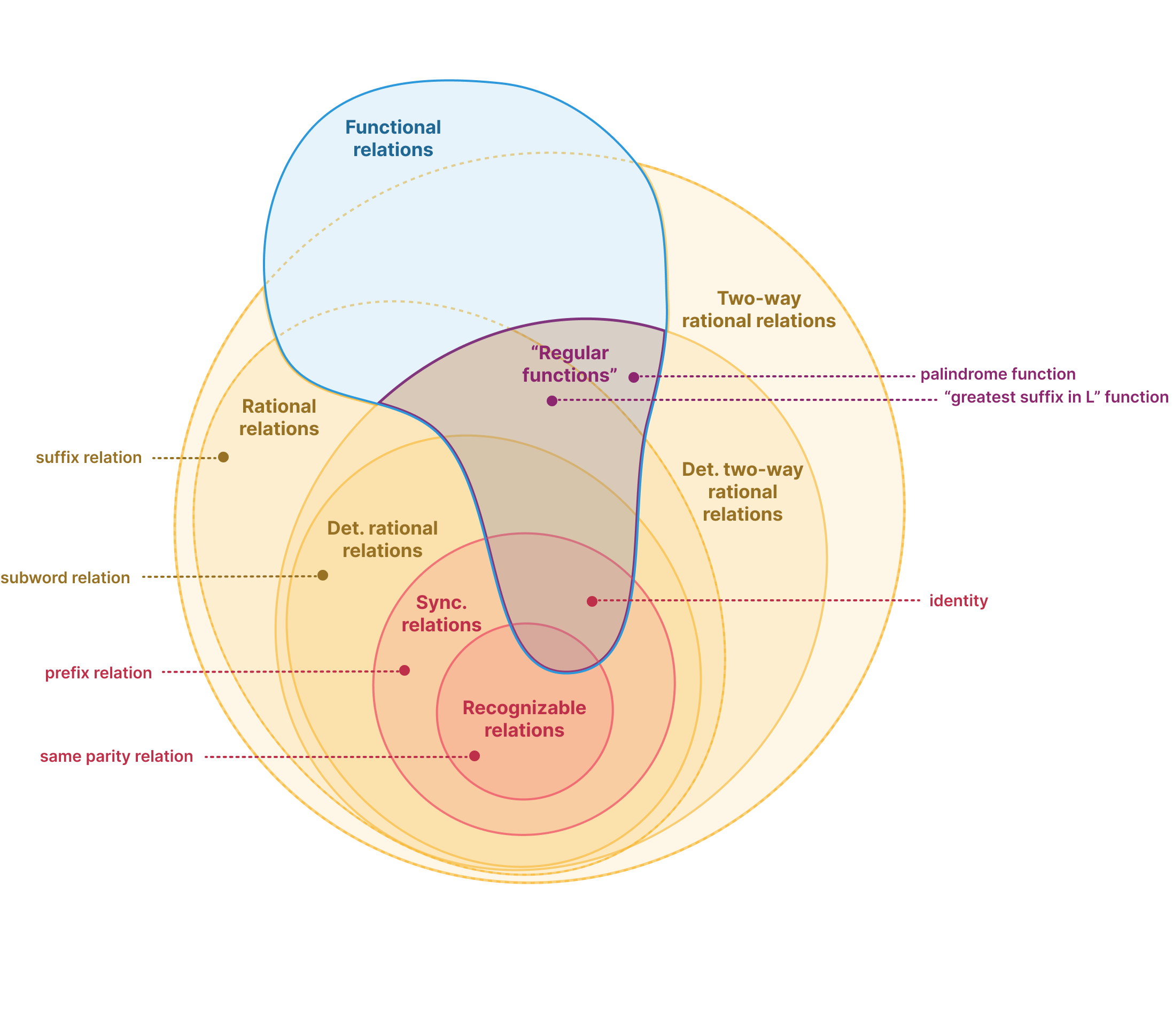}
	\caption{
		\label{fig:landscape-rationality} The landscape of rationality for binary relations.
		Dashed regions are empty: the intersection of
		functional relations and two-way rational relations
		collapses to regular functions by
		\cite[Theorem 22, p.~243]{EH2001transduction}.
	}
\end{figure}
The algebraic framework has been extended far beyond languages of finite words: let us cite amongst other Reutenauer's ``algèbre associative syntactique'' for weighted languages
\cite[Théorème I.2.1, p.~451]{reutenauer_series_1980} and their associated Eilenberg theorem \cite[Théorème III.1.1, p.~469]{reutenauer_series_1980};
for languages of $\omega$-words, Wilke's algebras and $\omega$-semigroups,
see \cite[\S II, pp.~75--131 \& \S VI, pp.~265--306]{perrin_infinite_2004};
more generally, for languages over countable linear orderings, see Carton, Colcombet \& Puppis' ``$\circledast$-monoids'' and ``$\circledast$-algebras''
\cite[\S 3, p.~7]{carton_algebraic_2018}.
A systemic approach has been recently developed using monads, see \Cref{apdx:monads}.
Non-linear structures are also suited to such an approach, see "eg" Bojańczyk \& Walukiewicz's 
forest algebras \cite[\S 1.3, p.~4]{bojanczyk_forest_2008} \cite[\S 5, p.~159]{bojanczyk_languages_2020}, or Engelfriet's hyperedge replacement algebras for graph languages
\cite[\S 2.3, p.~100]{courcelle_graph_2012} \cite[\S 6.2, p.~194]{bojanczyk_recognisable_2015}.
For relations over words ("aka" transductions), "recognizable 
relations" are exactly the ones recognized by monoid morphisms $\Sigma^* \times \Sigma^* \to M$ 
where $M$ is finite. This can be trivially generalized to show 
that a relation $\+R$ is a finite union of Cartesian products of languages in $\+V$ if, and only 
if, it is recognized by a monoid from $\B{V}$, the pseudovariety of monoids corresponding to
$\+V$, see \Cref{apdx:relations-recognizable}.
In 2023, Bojańczyk \& Nguy\smash{\~{\^e}}n 
managed to develop an algebraic structure called ``transducer semigroups'' for ``regular functions'' \cite[Theorem 3.2, p.~6]{bojanczyk_algebraic_2023}, an 
orthogonal class of relations to ours---see \Cref{fig:landscape-rationality}.

The counterpart of "$\+V$-relations" for rational relations---that we call here $\+V$-rational relations---was studied by Filiot, Gauwin \& Lhote \cite{filiot_logical_2019}: they show that if
$\+V$ has decidable membership, then ``$\+V$-rational transductions'' also have decidable membership
\cite[Theorem 4.10, p.~26]{filiot_logical_2019}.
``Rational transductions'' correspond in \Cref{fig:landscape-rationality} to the intersection of functional relations with rational relations: this class
is orthogonal to "synchronous relations",
but is included in the class of ``regular functions''.
A different problem---focussing more on the semantics of the transduction---, called ``$\+V$-continuity'' was studied by Cadilhac, Carton \& Paperman \cite[Theorem 1.3, p.~3]{cadilhac_continuity_2020}, although it has to be noted that their results only concern
a finite number of pseudovarieties.
\section{Preliminaries}
\label{sec:preliminaries}

\subsection{Automata \& Relations}

We assume familiarity with basic algebraic language theory over finite words, see \cite[\S 1, 2, 4, pp.~3--66 \& pp.~107--156]{bojanczyk_languages_2020} for a succinct and monad-driven approach, or \cite[\S I--XIV, pp.~3--247]{pin_mathematical_2022} for a more detailed presentation of the domain.
We also refer to \cite{SW21varieties} for a presentation on pseudovarieties.\footnote{``Pseudovarieties of \emph{foo}'' and ``varieties of finite \emph{foo}''---where \emph{foo} is "eg" ``groups''
or ``semigroups''---are used interchangeably in the literature.}
More precise pointers are given in \Cref{apdx:pointers-pin}.

A \AP""relation"" is a subset of $\Sigma^*\times\Sigma^*$,
where $\Sigma$ is an alphabet---"ie" a non-empty finite set.
We define its \AP""complement@@rel"" \AP$\intro*\negrel \+R$ as the "relation" $\{(u,v) \in \Sigma^*\times \Sigma^* \mid (u,v) \not\in \+R\}$.
Letting
$\intro*\SigmaPair \defeq
(\Sigma \times \Sigma) \,\cup\,
(\Sigma \times \{\pad\}) \,\cup\,
(\{\pad\} \times \Sigma)$, a ""synchronous automaton"" is a finite-state machine with initial states, final states, and
non-deterministic transitions labelled by elements of $\SigmaPair$.
We denote by \AP$\intro*\WellFormed$ the set of \AP""well-formed"" words over $\SigmaPair$ where the padding symbols are placed consistently, namely: if some padding symbol occurs on a tape/component, then the following symbols of this tape/component must all be padding symbols.
From this constraint, and since $\pair{\pad}{\pad} \not\in \SigmaPair$,
there can never be padding symbols on both tapes.

Note that elements of $\WellFormed$ are in natural bijection with $\Sigma^*\times\Sigma^*$---see
\Cref{fig:ex-sync-auto}.
The relation recognized by a "synchronous automaton" is the set of pairs $(u,v) \in \Sigma^*\times\Sigma^*$ such that their corresponding element in $\WellFormed$ is the label of
an accepting run of the automaton. We say that a relation is \AP""synchronous"" if it is recognized 
by such a machine.

\begin{remark}
	\label{rk:semantic}
	Crucially, in the semantics of "synchronous automata" we \emph{never}
	try to feed them inputs where the padding symbols are not consistent: for instance, while
	\[
		\pair{aab}{b\pad a},
		\text{ or }
		\pair{aba\pad}{a\pad\pad b}	
	\]
	are sequences in $(\SigmaPair)^*$, the behaviour of a "synchronous automaton"
	on such sequences is completely disregarded to define the relation it recognizes. 
\end{remark}

We can then reformulate the definition of the semantics of a "synchronous automaton",
to make the connection with "$\+V$-relations"---see the next subsection---explicit.

\begin{fact}
	\label{fact:synchronous-is-regular}
	Given a "synchronous automaton", its semantics as a "synchronous automaton"
	can be written as the intersection of its semantics as a classical automaton over $\SigmaPair$
	with $\WellFormed$.
\end{fact}

In particular a relation $\+R$ is "synchronous" if, and only if, it is a regular language when seen as a subset of $(\SigmaPair)^*$.

\subsection{Induced Relations}

Given a class $\+V$ of regular languages,
the class of \AP""$\+V$-relations"" over $\Sigma$ consists of all "relations"
of the form $L \cap \WellFormed$ for some $L \in \+V_{\SigmaPair}$---see \Cref{fig:projection}.\footnote{The notation $L \in \+V_{\SigmaPair}$ means that $L$ is a language over
the alphabet $\SigmaPair$. See \cite[introduction of \S{}XIII.1]{pin_mathematical_2022}
for why classes of regular languages are defined in such a way.}

\begin{figure}[tbp]
	\begin{center}
		\scalebox{1}{
		\begin{tikzpicture}[shorten >= 1pt, node distance = 2cm, on grid, baseline]
			\AP\small
			\node[state, initial left, accepting] (q0) {}; 
			\node[state] (q1) [right =of q0] {};
			\node[state, accepting] (q0') [below left = 1.5cm and 1cm of q0] {};
			\node[state, accepting] (q0'') [below right = 1.5cm and 1cm of q0] {};
			\node[state] (qbot) [below right = 1.5cm and 1cm of q0''] {};
			\path[->]
				(q0) edge[loop above] node[font=\scriptsize] {$\pair{a}{a}, \pair{b}{b}$} (q0) 
				(q0) edge[bend left=20] node[above=0pt, font=\scriptsize] {$\pair{a}{b},\pair{b}{a}$} (q1)
				(q1) edge[bend left=20] node[above=0pt, font=\scriptsize] {$\pair{a}{b},\pair{b}{a}$} (q0)
				(q1) edge[loop above] node[font=\scriptsize] {$\pair{a}{a}, \pair{b}{b}$} (q1)
				(q0) edge node[left] {$\pair{a}{\pad}, \pair{b}{\pad}$} (q0')
				(q0) edge node[right] {$\pair{\pad}{a}, \pair{\pad}{b}$} (q0'')
				(q0') edge[loop left] node[left] {$\pair{a}{\pad}, \pair{b}{\pad}$} (q0')
				(q0'') edge[loop right] node[right] {$\pair{\pad}{a}, \pair{\pad}{b}$} (q0'')
				(q0') edge[bend right=20] node[left] {$*$} (qbot)
				(q0'') edge node[right] {$*$} (qbot)
				(q1) edge[bend left=90] node[right] {$*$} (qbot)
				(qbot) edge[loop below] node[below] {$*$} (qbot);
		\end{tikzpicture}
		}
	\end{center}
	\caption{
		\label{fig:min-auto}
		Minimal (deterministic complete) ``classical'' automaton for 
		the binary relation
		of pairs $(u,v)$ such that the number of $a$'s in $u_1\hdots u_k$
		and in $v_1\hdots v_k$ are the same mod $2$, where $k = \min(|u|, |v|)$,
		seen as a language over $\SigmaPair$.
		Said otherwise, this is automaton rejects exactly all words in $(\SigmaPair)^*$ which (1)
		are not the valid encoding of a pair of words and (2) are the encoding of a pair
		which does not satisfy the property above.
		Each label $*$ is defined so that the automaton is deterministic and complete.
	} 
\end{figure}
For instance, if $\+V$ is the class of all regular languages, then by
\Cref{fact:synchronous-is-regular}, "$\+V$-relations" are exactly the "regular relations", "aka" "synchronous relations"!
However, because of
\Cref{rk:semantic}, the minimal automaton for a relation, seen as a language over $\SigmaPair$,
can be significantly more complex than a deterministic complete "synchronous automaton" recognizing it, see \Cref{fig:min-auto} in page~\pageref{fig:min-auto}---while the size blow-up is only polynomial, it breaks many of the structural properties of the automaton, such as the property of being a permutation automaton.

Note that if $\+R$ belongs to $\+V$ when $\+R$ is seen as a language over $\SigmaPair$,
then $\+R$ is a "$\+V$-relation".
The converse implication holds under some strong assumption on $\+V$ (\Cref{fact:a-triviality-on-trivial-relations}),
but is not true in general (\Cref{ex:group-languages}).

\begin{fact}
	\label{fact:a-triviality-on-trivial-relations}
	If $\+V$ is a class of languages closed under intersection and that contains $\WellFormed$, then
	a relation $\+R$ is a "$\+V$-relation" if, and only if, it belongs to $\+V$ when seen
	as a language over $\SigmaPair$.
\end{fact}

Classes of languages $\+V$ satisfying the previous assumption ("eg" first-order definable languages, piecewise-testable languages, etc.) are easy to capture when
it comes to "$\+V$-relations" since this class reduces to $\+V$-languages.
So, in the remaining of the paper, we will focus on classes $\+V$ which do not satisfy
the assumptions of \Cref{fact:a-triviality-on-trivial-relations}, such as group languages.

\begin{example}[Group relations]
	\label{ex:group-languages}
	If $\+V$ is the class
	of group languages, namely languages recognized by permutation automata\footnote{A permutation 
	automaton is a finite-state deterministic complete automaton whose transition functions are all 
	permutations of states.} or equivalently by a finite group, then we call
	"$\+V$-relations" \AP``""group relations""''. They can be characterized 
	as relations recognized by permutation "synchronous automata". For instance, the relation
	of \Cref{fig:ex-sync-auto} is a "group relation" as witnessed by the permutation synchronous automaton of \Cref{fig:ex-sync-auto}. Note however that it is not a group language, when seen as a language over $\SigmaPair$, since its minimal automaton over $\SigmaPair$ is not
	a permutation automaton, see \Cref{fig:min-auto} on \cpageref{fig:min-auto}.
\end{example}

\begin{fact}
	\label{fact:separability}
	Given a relation $\+R$ and a class $\+V$ of languages, the following are equivalent:
	\begin{enumerate}
		\item $\+R$ is a "$\+V$-relation";
		\item $\+R$ and $\negrel \+R$ are $\+V$-separable as languages over $\SigmaPair$,
		"ie" there is a language in $\+V$ which contains $\+R$ and does not intersect $\negrel \+R$.
	\end{enumerate}
\end{fact}

\begin{proof}
	By definition, see \Cref{fig:projection}, on page \pageref{fig:projection}.
\end{proof}

And so, if the $\+V$-separability problem is decidable, then the class of "$\+V$-relations"
is decidable. However, there are pseudovarieties $\+V$ with decidable membership but 
undecidable separability problem \cite[Corollary 1.6, p.~478]{rhode_pointlike_2011}.\footnote{The paper cited only claims undecidability of pointlikes, but it was noted in \cite[\S 1, pp.~1--2]{gool_pointlike_2019} that undecidability of the 2-pointlikes also holds, which is a problem 
equivalent to separability by \cite[Proposition 3.4, p.~6]{almeida_algorithmic_1999}.}
Moreover, some of these classes do not contain $\WellFormed$ \cite[Corollary 1.7, p.~478]{rhode_pointlike_2011}. But beyond 
this, even when a separation algorithm exists, it can be conceptually much harder than its
membership counterpart: for instance, deciding membership for group languages is trivial---it boils down to checking if a monoid is a group---, yet the decidability of the 
separation problem for group languages is considered to be one of the major results in semigroup theory:
it follows from Ash's infamous type II theorem \cite[Theorem 2.1, p.~129]{ash_inevitable_1991}, see \cite[Theorem 1.1, p.~3]{henckell_ashs_1991} for a presentation of the result in terms of pointlike sets, see also \cite[\S III, Theorem 8, p.~5]{place_group_2023} for an elegant automata-theoretic reformulation.

\section{Synchronous Algebras}
\label{sec:synchronous-algebras}

In this section, we introduce and study the ``elementary'' properties of "synchronous algebras".

\subsection{Types \& dependent Sets}

\subparagraph*{Motivation.} The axiomatization of a semigroup reflects the algebraic structure of
finite words: these objects can be concatenated, in an associative way---reflecting the linearity of 
words. Now observe that elements of $\WellFormed$ are still linear, but
not all words can be concatenated together: for instance, $\pair{a}{\pad}$
cannot be followed by $\pair{a}{b}$.
Formally, given two words $u, v \in \WellFormed$, to decide if $uv \in \WellFormed$
it is necessary and sufficient to know if the last pair of $u$ and first pair of $v$
consists of a pair of proper letters (denoted by \AP$\intro*\ll$), a pair of a proper letter and a blank/padding symbol (\AP$\intro*\lb$) or a pair of a blank/padding symbol and a proper letter (\AP$\intro*\bl$). This information is called the \AP""letter-type"" of an element of $\SigmaPair$.

We then define the \AP""type"" of a word of $(\SigmaPair)^+$ as the pair $(\alpha, \beta)$,
usually written $\alpha \to \beta$, of the "letter-types" of its first and last letters.
It is then routine to check that the possible types of "well-formed words" are
\[
	\intro*\types \defeq \big\{
		\ll\to\ll,\; \ll\to\lb,\; \lb\to\lb,\; \ll\to\bl,\; \bl\to\bl
	\big\}.
\]
For the sake of readability, we will write $\alpha$ instead of $\alpha \to \alpha$
for $\alpha \in \{\ll, \lb, \bl\}$.

One non-trivial point lies in the following innocuous question: what is the "type" of the empty word? Any "type" of $\types$ sounds like an acceptable answer. But then
it would be natural to say that the concatenation of $\pair{aaa}{aaa}$ of "type" $\ll$ with the empty word of "type" $\ll\to\lb$ should be $\pair{aaa}{aaa}$ of "type" $\ll\to\lb$. Automata-wise,
this would represent a sequence of transitions $\pair{a}{a}, \pair{a}{a}, \pair{a}{a}$ together with
the promise that the next transition would have a padding symbol on its second tape. But then, 
one has to formalize the idea that the two elements $\pair{aaa}{aaa}$ of type $\ll$ and $\ll\to\lb$
represent the same underlying pair of words of $\Sigma^*\times\Sigma^*$: this idea will be captured by what we call a "dependency relation". A more natural solution would be to simply introduce a
new type for the empty word (or to forbid it), but we show in \Cref{apdx:counterexample} that
the resulting notion of algebras cannot capture the property of being a "$\+V$-relation".

A \AP""$\types$-typed set"" (or \reintro{typed set} for short) consists of
a tuple $\?X = (X_\tau)_{\tau \in \types}$, where each $X_\tau$ is a set.
Instead of $x \in X_\tau$, we will often write \AP$\intro*\type{x}{\tau} \in \?X$.
A \AP""map between typed sets"" $\?X$ and $\?Y$ is a collection of functions
$X_\tau \to Y_\tau$ for each "type" $\tau$.
Similarly, a subset of $\?X$ is a tuple of subsets of $X_\tau$ for each "type" $\tau$.
To make the notations less heavy, we will often think of
"typed sets" as sets with type annotations rather than tuples, and ask that
all operators/constructions should preserve this "type".

\begin{definition}
	\label{def:dependency}
	A \AP""dependency relation"" over a "typed set" $\?X$ consists of
	a reflexive and symmetric relation $\intro*\dep$ over
	\AP$\intro*\disunion \?X \defeq \bigcup_{\tau \in \types} X_\tau \times \{\tau\}$,
	such that for all $\type{x}{\sigma}, \type{y}{\sigma} \in \?X$,
	if $\type{x}{\sigma} \dep \type{y}{\sigma}$,
	then $\type{x}{\sigma} = \type{y}{\tau}$.

	Crucially, we do not ask for "this relation@dependency" to be transitive---in some examples the "dependency relation" will be an equivalence relation, but not always (\Cref{ex:last_letter_is_a_if_big_diff}), and
	this non-transitivity is actually an important feature, motivated amongst other by the
	"syntactic congruence" and \Cref{coro:syntactic-congruence-is-syntactic-dependency}.

	A \AP""dependent set"" is a "$\types$-typed set" together with a
	"dependency relation" over it. A \AP""closed subset"" of a "dependent set" $\langle \?X, \dep \rangle$ is a subset $C \subseteq \?X$ such that for all
	$x, x' \in \?X$, if $x \dep x'$ then $x \in C \iff x' \in C$.\footnote{In other
	words, $C$ is a union of equivalence classes of the transitive closure of $\dep$.}
\end{definition}
\begin{figure}[htb]
	\begin{center}
		\includegraphics[width=\linewidth]{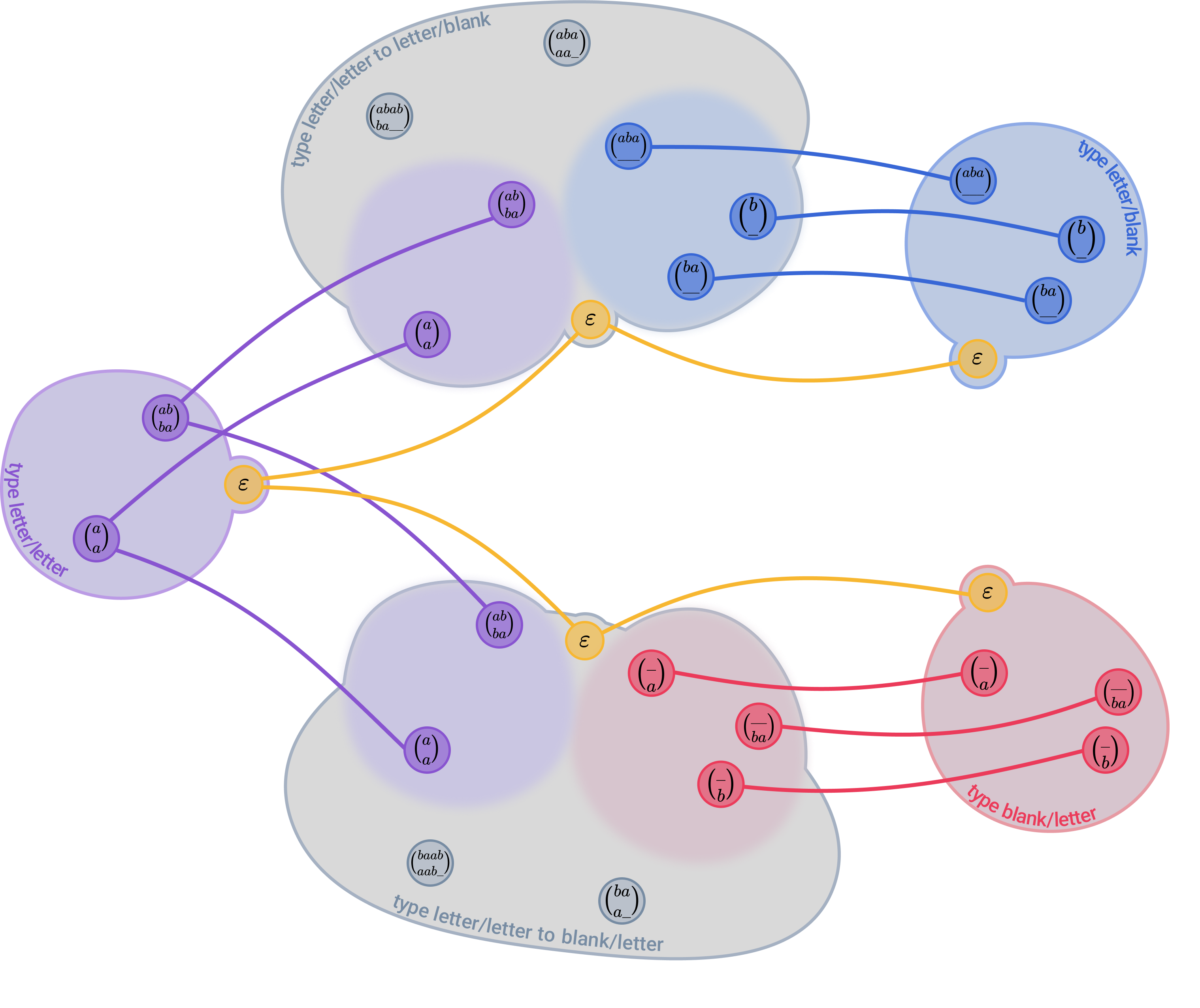}
	\end{center}
	\caption{
		\label{fig:free_algebras}
		Representation of the "dependent set" $\Sync\Sigma$ of "synchronous words".
		Coloured edges represent the "dependency relation", and self-loops are not drawn. 
	} 
\end{figure}
\begin{example}
	Given a finite alphabet $\Sigma$, let $\intro*\Sync{\Sigma}$ be\footnote{The index refers to the arity of the relations we are considering: here we focus on binary relations, but all constructions can be generalized to higher arities.} the "dependent set" of ""synchronous words"" defined by:
	\vspace{-1.25em}
	\begin{multicols}{2}
	\begin{itemize}
		\item $(\Sync{\Sigma})_{\ll} \defeq (\Sigma\times \Sigma)^*$,
		\item $(\Sync{\Sigma})_{\ll\to\lb} \defeq (\Sigma\times \Sigma)^*(\Sigma\times \pad)^*$,
		\item $(\Sync{\Sigma})_{\lb} \defeq (\Sigma\times \pad)^*$,
		\columnbreak
		\vfill
		\item $(\Sync{\Sigma})_{\ll\to\bl} \defeq (\Sigma\times \Sigma)^*(\pad\times \Sigma)^*$,
		\item $(\Sync{\Sigma})_{\bl} \defeq (\pad\times \Sigma)^*$.
	\end{itemize}
	\end{multicols}
	\vspace{-1em}
	Moreover, $\dep$ is the reflexive and symmetric closure of the relation that
	identifies $\type{u}{\ll}$ with $\type{u}{\ll\to\beta}$ for all
	$u \in (\Sigma\times\Sigma^*)$ and $\beta \in \{\lb,\bl\}$, 
	and $\type{u}{\ll\to\lb}$ with $\type{u}{\lb}$
	for $u \in (\Sigma\times\pad)^*$,
	and $\type{u}{\ll\to\bl}$ with $\type{u}{\bl}$
	for $u \in (\pad\times\Sigma)^*$.
	This structure is depicted in \Cref{fig:free_algebras}.
\end{example}


Given a "relation" $\+R \subseteq \Sigma^* \times \Sigma^*$, we denote by
$\AP\intro*\proj{\+R} = \{\type{(u,v)}{\tau} \mid \type{(u,v)}{\tau} \in \Sync{\Sigma} \text{ and } (u,v) \in \+R\}$
the "closed subset" of $\Sync{\Sigma}$ induced by $\+R$.

\begin{fact}
	The map $\+R \mapsto \proj{\+R}$ is a bijection between
	"relations" and "closed subsets" of $\Sync{\Sigma}$.
\end{fact}

\begin{proof}
	Let $f$ be the function which maps a closed subset $C$ of $\Sync{\Sigma}$ to
	$\{(u,v) \in \Sigma^* \times \Sigma^* \mid \type{(u,v)}{\tau} \in C \text{ for some $\tau \in \types$} \}$. It then follows that $f\circ \proj{-}$ (resp.~$\proj{f(-)}$) is the identity
	on subsets of $\Sigma^* \times \Sigma^*$ (resp.~"closed subsets" of $\Sync{\Sigma}$).
\end{proof}

\subsection{Synchronous Algebras}

One key property of "types" is that some of them can be concatenated to produce other types.
We say that two types $\sigma, \tau \in \types$ are \AP""compatible""
when there exists non-empty words $u,v \in \WellFormed$ of "type" $\sigma$
and $\tau$, respectively, such that $uv$ is "well-formed".
Said otherwise, $\alpha \to \beta$ is "compatible" with $\beta' \to \gamma$
if either $\beta = \beta'$ or $\beta = \ll$---indeed, for this last case
note that "eg" the concatenation of $\pair{aaa}{aaa}$ of type $\ll$ with
$\pair{\pad\pad}{aa}$ of type $\bl$ is "well-formed".
Lastly, if $\alpha\to \beta$ is "compatible with" $\beta'\to\gamma$,
we define their product as \AP $(\alpha\to\beta)\intro*\concattype(\beta'\to\gamma) \defeq \alpha \to \gamma$.
Note that this partial operation is associative, in the following sense: for $\rho,\sigma,\tau \in \types$, $(\rho\concattype\sigma) \concattype \tau$ is well-defined if and only if
$\rho\concattype(\sigma\concattype \tau)$ is well-defined, in which case both "types" are equal.
This implies that the notion of "compatibility" of types can be unambiguously
lifted to finite lists of "types" $\tau_1,\hdots,\tau_n$.

\begin{definition}
	\label{def:synchronous-algebra}
	A \AP""synchronous algebra"" $\langle \?A, \cdot, \dep \rangle$ consists of a "dependent set" $\langle \?A, \dep \rangle$ together
	with a partial binary operation $\cdot$ on $\?A$, called \AP""product"" such that:
	\begin{itemize}
		\item for $\type{x}{\sigma}, \type{y}{\tau} \in \?A$,
			$\type{x}{\sigma} \cdot \type{y}{\tau}$ is defined "iff" $\sigma$ and $\tau$
			are "compatible",
		\item \emph{associativity:} for all $\type{x}{\rho}, \type{y}{\sigma}, \type{z}{\tau} \in \?A$, if $\rho, \sigma, \tau$ are "compatible":
		\[
			(\type{x}{\rho} \cdot \type{y}{\sigma}) \cdot \type{z}{\tau} 
			= \type{x}{\rho} \cdot (\type{y}{\sigma} \cdot \type{z}{\tau}),
		\]
		\item \emph{``monotonicity'':} for all $\type{x}{\sigma}, \type{x'}{\sigma'}, \type{y}{\tau} \in \?A$, if $\type{x}{\sigma} \dep \type{x'}{\sigma'}$ and
		both $\sigma,\tau$ and $\sigma', \tau$ are "compatible", then
		$\type{x}{\sigma}\cdot \type{y}{\tau} \dep \type{x'}{\sigma'}\cdot \type{y}{\tau}$,
		and dually if $\tau,\sigma$ and $\tau, \sigma'$ are "compatible", then 
		$\type{y}{\tau}\cdot \type{x}{\sigma} \dep \type{y}{\tau}\cdot \type{x'}{\sigma'}$,
		\item \emph{units:} for each type $\tau$ there is an element $\type{1}{\tau} \in \?A$ 
			such that for any $\type{x}{\sigma} \in \?A$, then
			$\type{1}{\tau} \cdot \type{x}{\sigma} \dep \type{x}{\sigma}$
			if $\tau$ and $\sigma$ are "compatible",
			and $\type{x}{\sigma} \cdot \type{1}{\tau} \dep \type{x}{\sigma}$
			if $\sigma$ and $\tau$ are "compatible",
			and moreover, $\type{1}{\ll\to\beta} = \type{1}{\ll}\cdot \type{1}{\beta}$
			for $\beta \in \{\lb,\bl\}$.
	\end{itemize}
\end{definition}

Note in particular that for any type $\tau \in \{\ll,\lb,\bl\}$,
then $\type{1}{\tau} \cdot \type{x}{\tau} \dep \type{x}{\tau}$
but since $\type{1}{\tau} \cdot \type{x}{\tau}$ has "type" $\tau$ and $\dep$ is a
"dependency relation", then $\type{1}{\tau} \cdot \type{x}{\tau} = \type{x}{\tau}$.
This implies in particular that restricting $\langle \?A, \cdot \rangle$
to a type $\ll$, $\lb$ or $\bl$ yields a monoid.
These are called the three \AP""underlying monoids"" of $\?A$.
The canonical example of "synchronous algebras" is "synchronous words" $\Sync\Sigma$ under 
concatenation. Its "underlying monoids" are $(\Sigma\times\Sigma)^*$,
$(\Sigma\times\{\pad\})^*$ and $(\{\pad\}\times\Sigma)^*$.

\begin{fact}
	\label{fact:closed-subset-units}
	Any "closed subset" of $\?A$ either contains all units, or none of them.
\end{fact}

\begin{proof}
	From $\type{1}{\ll \to \lb} = \type{1}{\ll}\cdot \type{1}{\lb}$ we have
	$\type{1}{\ll} \dep \type{1}{\ll \to \lb}$
	and $\type{1}{\ll \to \lb} \dep \type{1}{\lb}$.
	By symmetry between $\lb$ and $\bl$, we also have
	$\type{1}{\ll} \dep \type{1}{\ll\to\bl}$ and $\type{1}{\ll\to\bl} \dep \type{1}{\bl}$.
	Hence, if a "closed subset" of $\?A$ contains at least one unit, then it must contain them all.
\end{proof}

Note that the "product" induces a monoid left (resp.~right) action of the
"underlying monoid" $\?A_{\ll}$ (resp.~$\?A_{\lb}$) on the set $\?A_{\ll\to\lb}$.
Moreover, $\type{x}{\ll} \mapsto \type{x}{\ll} \cdot \type{1}{\lb}$
identifies any element of "type" $\ll$ with an element of "type" $\ll\to\lb$.
Over $\Sync{\Sigma}$, these identifications are injective, but it need not be the
case in general. Note also that in general,
$\type{x}{\ll} \cdot \type{1}{\ll\to\lb} = \type{x}{\ll} \cdot \type{1}{\ll}
\cdot \type{1}{\lb} = \type{x}{\ll} \cdot \type{1}{\lb}$.

\begin{remark}
	There exists a monad over the category of "dependent sets" whose
	"Eilenberg-Moore algebras" exactly correspond to
	"synchronous algebras", see \Cref{apdx:monads}.
\end{remark}

\AP""Morphisms of synchronous algebras"" are defined naturally as
maps that preserve the "type", units, the "product" and the "dependency relation".

\subparagraph*{Free algebras}
$\Sync\Sigma$ is free in the sense that
for any "synchronous algebra" $\?A$, there is a natural bijection between
"synchronous algebra morphisms" $\Sync\Sigma \to \?A$
and "maps of typed sets" $\SigmaPair \to \?A$.
Said otherwise, "synchronous algebra morphisms" are uniquely defined by
their value on $\SigmaPair$.

\subsection{Recognizability}

Given a "synchronous algebra" $\?A$, a "morphism@@sync" $\phi\colon \Sync{\Sigma} \to \?A$
and a "closed subset" $\Acc \subseteq \?A$ called ``accepting set'',
we say that $\langle \phi, \?A, \Acc \rangle$
\AP""recognizes@@sync"" a relation $\+R \subseteq \Sigma^* \times \Sigma^*$
when $\proj{\+R} = \phi^{-1}[\Acc]$.
We extend the notion of "recognizability@@sync" to $\langle \phi, \?A\rangle$
or to simply $\?A$ by existential quantification over the missing elements in the tuple $\langle \phi, \?A, \Acc \rangle$.

\subparagraph*{Synchronous algebra induced by a monoid}
A monoid morphism $\phi\colon (\SigmaPair)^* \to M$
naturally ""induces@@sync"" a "synchronous algebra morphism"
\AP$\intro*\inducedmor{\phi} \colon \Sync\Sigma \to \intro*\inducedalg{M}$,
where:
\begin{itemize}
	\item $\inducedalg{M}$ has for every "type" $\tau$ a copy of $M$,
	and $\dep$ is $\{(\type{x}{\sigma}, \type{x}{\tau}) \mid x \in M, \sigma, \tau \in \types\}$,
	\item for all $\type{x}{\sigma}, \type{y}{\tau} \in \inducedalg{M}$
	with "compatible type", $\type{x}{\sigma} \cdot \type{y}{\tau} \defeq \type{(x\cdot y)}{\sigma\concattype\tau}$,
	\item $\inducedmor{\phi}\pair{a}{b} \defeq \type{\bigl(\phi\pair{a}{b}\bigr)}{\ll}$,
		$\inducedmor{\phi}\pair{a}{\pad} \defeq \type{\bigl(\phi\pair{a}{\pad}\bigr)}{\lb}$,
		and $\inducedmor{\phi}\pair{\pad}{a} \defeq \type{\bigl(\phi\pair{\pad}{a}\bigr)}{\bl}$.
\end{itemize}
The algebra simply duplicates $M$ as many times as needed and identifies
two elements together when they originated from the same element of $M$.
\begin{fact}
	\label{fact:induced-morphism}
	If $\phi$ recognizes $\+R$ for some "relation" $\+R \subseteq \Sigma^* \times \Sigma^*$ seen as a language over $\SigmaPair$, then $\inducedmor{\phi}$ "recognizes@@sync" $\+R$.
\end{fact}

\subparagraph*{Consolidation of a synchronous algebra}
Given a "synchronous algebra morphism" $\phi\colon \Sync\Sigma \to \?A$,
define its \AP""consolidation""\footnote{Named by analogy with Tilson's construction \cite[\S 3, p.~102]{tilson_categories_1987}.} as the semigroup morphism
$\intro*\consol{\phi}\colon (\SigmaPair)^* \to \reintro*\consol{\?A}$, where
$\consol{\?A}$ is the monoid obtained from
$\disunion{\?A}$ by first merging units,
by adding a zero (denoted by $0$), and extending $\cdot$
to be a total function by letting all missing products equal $0$,
and $\consol{\phi}$ sends a word $u \in (\SigmaPair)^*$ to%
\vspace{-1.25em}
\begin{multicols}{2}%
	\begin{itemize}%
		\item $0$ if $u$ is not "well-formed",
		\item $\phi(\type{u}{\ll})$ if $u \in (\Sigma\times \Sigma)^*$,
		\item $\phi(\type{u}{\lb})$ if $u \in (\Sigma\times \pad)^+$,
		\item $\phi(\type{u}{\bl})$ if $u \in (\pad\times \Sigma)^+$,
		\item $\phi(\type{u}{\ll\to\lb})$ if $u \in (\Sigma\times \Sigma)^+(\Sigma\times \pad)^+$,
		\item $\phi(\type{u}{\ll\to\bl})$ if $u \in (\Sigma\times \Sigma)^+(\pad\times \Sigma)^+$.
	\end{itemize}
\end{multicols}
\vspace{-1.25em}
Note that this operation disregards the "dependency relation" of $\?A$.
\begin{fact}
	\label{fact:consolidation}
	If $\phi$ "recognizes@@sync" some "relation" $\proj{\+R}$,
	then $\consol{\phi}$ recognizes $\+R$, when seen as
	a language over $\SigmaPair$.
\end{fact}

The following result follows from \Cref{fact:induced-morphism,fact:consolidation,fact:synchronous-is-regular}.
\begin{proposition}
	\label{prop:synchronous-iff-finite}
	A "relation" is "synchronous" if and only if it is "recognized@@sync"
	by a finite "synchronous algebra".
\end{proposition}

Let us continue with a slightly less trivial example of "algebra@synchronous algebra".

\begin{example}[Group relations: {\Cref{ex:group-languages}, cont'd.}]
	\label{ex:algebra-Zpq}
	Fix $p,q \in \N_{>0}$. Let $\?Z_{p,q}$ denote the algebra whose "underlying monoids" are:
	\begin{itemize}
		\item the trivial monoid $(0,+)$ for "type" $\ll$,
		\item the cyclic monoid $(\Z/p\Z,+)$ for "type" $\lb$,
		\item the cyclic monoid $(\Z/q\Z,+)$ for "type" $\bl$.
	\end{itemize}
	Moreover, the sets $Z_{\ll\to\lb}$ and $Z_{\ll\to\bl}$ are defined
	as $\Z/p\Z$ and $\Z/q\Z$, respectively.
	The product is addition---we identify $\type{0}{\ll}$ with the zero of $\Z/p\Z$ and
	of $\Z/q\Z$. We denote by $\bar k$ the equivalence class of $k\in \Z$ in
	$\Z/n\Z$ when $n$ is clear from context.
	The "dependency relation" identifies (1) all units together and (2) $\type{x}{\sigma}$ with $\type{1}{\tau} \cdot \type{x}{\sigma}$ and $\type{x}{\sigma} \cdot \type{1}{\tau}$ when the "types" are "compatible".
	
	Let $\phi\colon \Sync\Sigma \to \?Z_{p,q}$ be the "synchronous algebra morphism"
	defined by
	\[
		\phi \pair{a}{b} \defeq  \type{\bar 0}{\ll}, 
		\quad
		\phi \pair{a}{\pad} \defeq \type{\bar 1}{\lb},
		\quad
		\phi \pair{\pad}{a} \defeq \type{\bar 1}{\bl}
		\quad\text{and}\quad
		\phi(\type{\varepsilon}{\tau}) \defeq \type{\bar 0}{\tau}
		\text{ for $\tau \in \types$}.
	\]
	This "morphism@@sync" recognizes any "relation" of the form 
	\begin{align*}
		\+R^{I,J} \defeq \big\{
			(u,v) \;\big\vert\; & |u| > |v| \text{ and } (|u| - |v| \bmod{p}) \in I, \text{ or} \\
			& |u| < |v| \text{ and } (|v| - |u| \bmod{q}) \in J.
		\hphantom{\text{ or}}\big\},
	\end{align*}
	where $I \subseteq \Z/p\Z$ and $J \subseteq \Z/q\Z$ are such that $\bar 0 \not\in I$
	and $\bar 0 \not\in J$. This last condition is necessary because the accepting set
	has to be a "closed subset" of $\?Z_{p,q}$: if $\bar 0$ was in $I$, then we would need
	$\bar 0 \in J$, but also to add $\type{\bar 0}{\ll}$ to the accepting set: this would "recognize@@sync"
	\begin{align*}
		\big\{
			(u,v) \;\big\vert\; & |u| > |v| \text{ and } (|u| - |v| \bmod{p}) \in I, \text{ or} \\
			& |u| < |v| \text{ and } (|v| - |u| \bmod{q}) \in J, \text{ or } |u| = |v| \big\}.
	\end{align*}
	Note also that all "relations" $\+R^{I,J}$ with $\bar 0 \not\in I$
	and $\bar 0 \not\in J$ are "group relations": letting $G$ be the group $\Z/p\Z \times \Z/q\Z$,
	$\+R$ can be written as $\WellFormed\cap \psi^{-1}[I \times \{0\} \cup \{0\} \times J]$ where $\psi \colon (\SigmaPair)^* \to G$ is the monoid morphism defined by $\psi\pair{a}{b} \defeq 
	(\bar 0, \bar 0)$, $\psi\pair{a}{\pad} \defeq (\bar 1, \bar 0)$ and $\psi\pair{\pad}{a} \defeq
	(\bar 0, \bar 1)$.
\end{example}

\subsection{Syntactic Morphisms \& Algebras}
\label{sec:syntactic-morphism}

\AP\phantomintro{\SyntSAM}\phantomintro{\SyntSA}\phantomintro{syntactic synchronous algebra morphism}\phantomintro{syntactic synchronous algebra}\phantomintro(sync){Syntactic morphism theorem}
\vspace{-2em}
\begin{restatable}[\reintro(sync){Syntactic morphism theorem}]{lemma}{syntacticmorphismtheorem}
	\label{lem:syntactic-morphism-theorem}
	For each relation $\+R$, there exists a surjective "synchronous algebra morphism"
	\[\reintro*\SyntSAM{\+R}\colon \Sync{\Sigma} \surj \reintro*\SyntSA{\+R}\]
	that "recognizes@@sync" $\+R$ and is such that for any other
	surjective "synchronous algebra morphism"
	$\phi\colon \Sync{\Sigma} \surj \?B$
	"recognizing@@sync" $\+R$, there exists a
	"synchronous algebra morphism" $\psi \colon \?B \surj \SyntSA{\+R}$
	such that the diagram
	\begin{center}
		\begin{tikzcd}[ampersand replacement=\&]
			\Sync{\Sigma} 
				\arrow[twoheadrightarrow]{r}{\SyntSAM{\+R}}
				\arrow[twoheadrightarrow, swap]{dr}{\phi}
			\& \SyntSA{\+R} \\
			\& \?B, \arrow[twoheadrightarrow, swap]{u}{\psi}
		\end{tikzcd}
	\end{center}
	commutes.
	The objects $\SyntSAM{\+R}$ and $\SyntSA{\+R}$ are called the \reintro{syntactic synchronous algebra morphism}
	and \reintro{syntactic synchronous algebra} of $\+R$, respectively.
	Moreover, these objects are unique up to "isomorphisms" of the "algebra@@sync".
\end{restatable}

\begin{corollary}[of \Cref{prop:synchronous-iff-finite,lem:syntactic-morphism-theorem}]
	A "relation" is "synchronous" if and only if its "syntactic synchronous algebra"
	is finite.
\end{corollary}

The proof of \Cref{lem:syntactic-morphism-theorem}---see \Cref{apdx-proof:lem:syntactic-morphism-theorem}---relies,
as in the case of monoids, on the notion of congruence.

Given a "synchronous algebra" $\langle \?A, \dep, \cdot \rangle$,
a \AP""congruence@@sync"" is any reflexive, symmetric relation $\asympbar$ over $\?A$ which is
coarser than $\dep$, and which is \AP""locally transitive"",
meaning that for all $\type{x}{\sigma}, \type{x'}{\sigma},
\type{y}{\tau}, \type{y'}{\tau} \in \?X$,
if $\type{x'}{\sigma} \asympbar \type{x}{\sigma}$,
$\type{x}{\sigma} \asympbar \type{y}{\tau}$ and
$\type{y}{\tau} \asympbar \type{y'}{\tau}$,
then $\type{x'}{\sigma} \asympbar \type{y'}{\tau}$.\footnote{In particular, it implies that
$\asympbar$ is transitive when restricted to elements of the same type.}

The \AP""quotient structure"" $\intro*\quotient{\?A}{\asympbar}$
of $\?A$ by a "congruence@@sync" $\asympbar$ is defined as follows:
\begin{itemize}
	\item its underlying "typed set" consists of the
		equivalence classes of $\?A$ under the equivalence relation
		$\{(\type{x}{\sigma}, \type{y}{\sigma}) \mid \type{x}{\sigma} \asympbar \type{y}{\sigma}\}$,
		such a class being abusively denoted by \AP$\intro*\equivclass{x}{\asympbar}$,
	\item its "product" is the "product" induced by $\?A$, in the sense
		that $\equivclass{x}{\asympbar}\cdot \equivclass{y}{\asympbar} \defeq \equivclass{xy}{\asympbar}$, and
	\item its "dependency relation" is the relation induced by $\asympbar$,
		"ie" $\equivclass{x}{\asympbar}\dep \equivclass{y}{\asympbar}$ whenever $x \asympbar y$,
	\item its units are defined as the equivalence classes of the units of $\?A$.
\end{itemize}
Moreover, $x \mapsto \equivclass{x}{\asympbar}$ defines a surjective
"morphism of synchronous algebras" from $\?A$ to $\quotient{\?A}{\asympbar}$.

Given a "synchronous algebra" $\langle \?A, \dep, \cdot \rangle$
and a "closed subset" $C \subseteq \?A$, 
we define a "congruence@@sync" \AP$\intro*\congr{C}$, called
\AP""syntactic congruence"" of $C$ over $\?A$ by letting
$\type{a}{\sigma} \congr{C} \type{b}{\tau}$ when for all $x, y \in \?A$
\begin{itemize}
	\item if both
		$x \type{a}{\sigma} y$
		and $x \type{b}{\tau} y$ are defined,
		then $x \type{a}{\sigma} y \in C$
		"iff" $x \type{b}{\tau} y \in C$, and
	\item if both
		$x \type{a}{\sigma}$
		and $x \type{b}{\tau}$ are defined,
		then $x \type{a}{\sigma} \in C$
		"iff" $x \type{b}{\tau} \in C$, and
	\item if both $\type{a}{\sigma} y$
		and $\type{b}{\tau} y$ are defined,
		then $\type{a}{\sigma} y \in C$
		"iff" $\type{b}{\tau} y \in C$.
\end{itemize}
It is routine to check that the "syntactic congruence" is indeed a "congruence@@sync".
For instance, to prove that $\congr{C}$ is coarser than $\dep$, observe that
if $\type{a}{\sigma} \dep \type{b}{\tau}$, then
for all $x, y$ "st" both $x \type{a}{\sigma} y$
and $x \type{b}{\tau} y$ are defined, then $x \type{a}{\sigma} y \dep x \type{b}{\tau} y$,
and since $C$ is a "closed subset" of $\?A$, $x \type{a}{\sigma} y \in C$ "iff"
$x \type{b}{\tau} y \in C$. The other two conditions are proven in the same fashion.
Note however that while the relation is "locally transitive", it is not
transitive in general.

When $\+R \subseteq \Sigma^* \times\Sigma^*$ is a "relation", we abuse
the notation and write $\congr{\+R}$ to denote
the "syntactic congruence" $\congr{\proj{\+R}}$ of $\proj{\+R}$
in $\Sync\Sigma$.
The existence of the "syntactic morphism@@sync" then follows from
the next proposition, proven in \Cref{apdx-proof:lem:syntactic-morphism-theorem}.

\begin{restatable}{proposition}{propquotientbysyntacticcongruenceissyntacticalgebra}
	\label{prop:quotient-by-syntactic-congruence-is-syntactic-algebra}
	Let $\phi\colon \Sync\Sigma \surj \?A$ be a surjective "synchronous algebra morphism"
	that "recognizes@@sync" $\+R$, say $\proj{\+R} = \phi^{-1}[\Acc]$
	for some "closed subset" $\Acc \subseteq \?A$, then 
	\begin{center}
		\begin{tabular}{rccc}
			$\quotient{\phi}{\congr{\Acc}}\colon$
			& $\Sync\Sigma$
			& $\surj$
			& $\quotient{\?A}{\congr{\Acc}}$\\
			& $u$
			& $\mapsto$
			& $\equivclass{\phi(u)}{\congr{\Acc}}$
		\end{tabular}		
	\end{center}
	is the "syntactic morphism@@sync" of $\+R$.
\end{restatable}

\begin{corollary}
	\label{coro:syntactic-congruence-is-syntactic-dependency}
	In the "syntactic synchronous algebra" $\SyntSA{\+R}$, the "syntactic
	congruence" $\congr{\Acc}$ and the "dependency relation" $\dep$
	coincide. 
\end{corollary}

\begin{proof}
	By \Cref{prop:quotient-by-syntactic-congruence-is-syntactic-algebra}
	applied to the "syntactic morphism@@sync",
	$x \mapsto \equivclass{x}{\congr{\Acc}}$
	is an isomorphism from $\SyntSA{\+R}$ to $\quotient{\SyntSA{\+R}}{\congr{\Acc}}$.
	Hence, $\equivclass{x}{\congr{\Acc}} \dep \equivclass{y}{\congr{\Acc}}$
	in $\quotient{\SyntSA{\+R}}{\congr{\Acc}}$
	"iff" $x \dep y$ in $\SyntSA{\+R}$, for all $x,y \in \SyntSA{\+R}$.
	But then, the "dependency relation" $\dep$ of
	$\quotient{\SyntSA{\+R}}{\congr{\Acc}}$ is, by definition,
	such that $\equivclass{x}{\congr{\Acc}} \dep \equivclass{y}{\congr{\Acc}}$
	"iff" $x \congr{\Acc} y$.
	Putting both equivalences together, we get that
	$x \congr{\Acc} y$ "iff" $x \dep y$ for all $x,y \in \SyntSA{\+R}$.
\end{proof}

We provide in \Cref{ex:last_letter_is_a_if_big_diff}
a simple example of "syntactic synchronous algebra" whose "dependency relation"
is not an equivalence relation.  



\subparagraph*{Boolean operations}
\AP Given two "synchronous algebras" $\?A$ and $\?B$, define their \emph{Cartesian product}
$\?A \times \?B$ by taking, for each type $\tau$, the Cartesian product $A_\tau \times B_\tau$. 
Units, "product" are defined naturally, and the "dependency relation" is defined
by taking the conjunction over each component.
Then $\negrel \+R$ is "recognized@@sync" by $\?A$, and
$\+R \cup \+S$ and $\+R \cap \+S$ are "recognized@@sync" by $\?A \times \?B$. 


\section{The Lifting Theorem \& Pseudovarieties}
\label{sec:lifting-theorem}

\subsection{Elementary Formulation}

\begin{example}[Group relations: {\Cref{ex:algebra-Zpq} cont'd}]
	\label{ex:charac-group-relation-monoids}
	We want to decide when the relation
	\begin{align*}
		\+R^{I,J} \defeq \big\{
			(u,v) \;\big\vert\; & |u| > |v| \text{ and } (|u| - |v| \bmod{p}) \in I, \text{ or} \\
			& |u| < |v| \text{ and } (|v| - |u| \bmod{q}) \in J.
		\hphantom{\text{ or}}\big\}	
	\end{align*}
	from \Cref{ex:algebra-Zpq} is a \AP"group relation".
	By definition this happens if and only if there exists a finite group $G$,
	together with a monoid morphism $\phi\colon (\SigmaPair)^* \to G$ and a subset
	$\Acc \subseteq G$ "st" $\forall u \in \WellFormed$, $u\in \+R^{I,J}$ "iff" $\phi(u) \in \Acc$.
	We claim:
	\begin{equation}
		\+R^{I,J} \text{ is a "group relation"}\quad\text{"iff"}\quad
		\big(
			\bar 0 \not\in I \text{ and } \bar 0 \not\in J
		\big).
		\tag{{\Large\adforn{10}}}
	\end{equation}

	The right-to-left implication was shown in \Cref{ex:algebra-Zpq}.
	We prove the implication from left to right:
	let $n$ be the order of $G$ so that $x^n = 1$ for all $x\in G$. In particular, we have:
	$\phi\bigl(\pair{a}{\pad}^{pqn}\bigr) = 1 = \phi\bigl(\pair{a}{a}^{pqn}\bigr)$.
	Since $\phi\bigl(\pair{a}{a}^{pqn}\bigr) \not\in \+R^{I,J}$, it follows that 
	$\pair{a}{\pad}^{pqn} \not\in \+R^{I,J}$
	"ie" $\bar 0 \not\in I$. Also, $\bar 0 \not\in J$ by symmetry, which concludes the proof.
\end{example}

\AP\phantomintro{Lifting Theorem}
Even more generally, we can decide if a "relation" $\+R$ is a "group relation"
by simply looking at the "syntactic synchronous algebra" of $\+R$.

\liftingtheoremmonoids

See the proof in \Cref{apdx-proof:thm:lifting-theorem-monoids}.

\begin{remark}
	In light of \Cref{thm:lifting-theorem-monoids},
	one can wonder whether the notion of "synchronous algebra" is necessary to
	characterize "$\+V$-relations", or if it is enough to look at the languages corresponding to
	the "underlying monoids".
	Said otherwise, is the membership of $\+R$ in the class of "$\+V$-relations" uniquely
	determined by the regular languages $\+R \cap (\Sigma\times\Sigma)^*$,
	$\+R \cap (\Sigma\times\{\pad\})^*$ and $\+R \cap (\{\pad\}\times\Sigma)^*$?
	Unsurprisingly, "synchronous algebras" are indeed necessary, as 
	there are relations $\+R$ such that:
	\begin{equation}
		\label{eq:naive-characterization} \tag{\adforn{12}}
		\proj{\+R} \cap (\Sigma\times \Sigma)^* \in \+V_{\Sigma\times \Sigma}, \quad
		\proj{\+R} \cap (\Sigma\times \pad)^* \in \+V_{\Sigma\times \pad}\quad\text{and}\quad
		\proj{\+R} \cap (\pad\times \Sigma)^* \in \+V_{\pad\times \Sigma},
	\end{equation}
	but $\+R$ is \emph{not} a "$\+V$-relation". This can happen even if
	$\+V$ is the "$\ast$-pseudovariety@@reglang" of all regular languages: for instance for the
	"relation"
	\[
		\+R \defeq \{(u,v) \mid |u| > |v| > 0 \text{ and } |u| - |v| \text{ is prime}\}.
	\]
	Notice that there is a subtle but crucially important difference between
	\eqref{eq:naive-characterization} and the second item of the "Lifting Theorem":
	while the "underlying monoids" of a "synchronous algebra" $\?A$ "recognizing@@sync" $\+R$
	only accept words of the form $(\Sigma\times \Sigma)^*$, $(\Sigma\times \pad)^*$
	or $(\pad\times \Sigma)^*$, elements of $(\Sigma\times \Sigma)^+(\Sigma\times \pad)^+$
	or $(\Sigma\times \Sigma)^+(\pad\times \Sigma)^+$ influence the "underlying monoids" of $\?A$
	via the axioms of "synchronous algebras".
\end{remark}

Also, note that the existence the "Lifting Theorem" follows from the careful
definition of "synchronous algebras": more naive definitions of these algebras
simply cannot characterize "$\+V$-relations", see \Cref{apdx:counterexample}.

From \Cref{thm:lifting-theorem-monoids}
and the implicit fact that all our constructions are effective,
we obtain a decidability (meta-)result for
"$\+V$-relations".
\begin{corollary}
	\label{coro:decidability}
	The class of "$\+V$-relations" has decidable membership if, and only if, $\+V$ has decidable membership.
\end{corollary}

For instance, a "relation" is a "group relation" if, and only if, all
"underlying monoids" of its "syntactic synchronous algebra" are
groups. 

\subsection{Pseudovarieties of Synchronous Relations}
\label{sec:varieties}

We introduce the notion of "pseudovariety of synchronous algebras" 
and "$\ast$-pseudovariety of synchronous relations". We show an "Eilenberg correspondence@@sync" between these two notions. We then reformulate the "Lifting Theorem"
to show that any "Eilenberg correspondence@@lang" between monoids and regular languages
lifts to an "Eilenberg correspondence@@sync" between "synchronous algebras" and "synchronous relations".

Say that a "synchronous algebra" $\?A$ is a \AP""quotient@@sync"" of $\?B$
when there exists a surjective "synchronous algebra morphism" from $\?B$ to $\?A$.
A ""subalgebra@@sync"" of $\?B$ is any "closed subset" of $\?B$ closed under "product"
and containing the units.
We then say that "synchronous algebra" $\?A$ \AP""divides@@sync"" $\?B$
when $\?A$ is a "quotient@@sync" of a "subalgebra@@sync" of $\?B$.

Observe that $\Sync\Sigma$ admits the following property:
elements of type $\ll\to\lb$ and $\ll\to\bl$ are generated by the "underlying monoids".
Since "syntactic synchronous algebras" are homomorphic images of $\Sync\Sigma$, they also
satisfy this property. In general, we say that a "synchronous algebra" $\?A$ is \AP""locally 
generated@@sync"" if every element of type $\ll\to\lb$ (resp.~$\ll\to\bl$)
can be written as the product of an element of type $\ll$ with an element of type $\lb$ (resp.~$\bl$).

A \AP""pseudovariety of synchronous algebras"" is any class $\B{V}$
of "locally generated@@sync" finite "synchronous algebras" closed under
\begin{itemize}
	\item \emph{finite product:} if $\?A, \?B \in \B{V}$ then $\?A \times \?B \in \B{V}$,
	\item \emph{division:} if some finite "locally generated" algebra $\?A$ "divides@@sync" $\?B$ for some $\?B \in \B{V}$, then $\?A \in \B{V}$.
\end{itemize}

Because of \Cref{lem:syntactic-morphism-theorem}, a "synchronous relation" is "recognized@@sync"
by a finite synchronous algebra of a "pseudovariety@@syncalg" $\B{V}$ "iff"
its "syntactic synchronous algebra" belongs to $\B{V}$.

A \AP""$\ast$-pseudovariety of synchronous relations"" is a function $\+V \colon \Sigma \mapsto \+V_\Sigma$
such that for any finite alphabet $\Sigma$, $\+V_\Sigma$ is a set of "synchronous relations" over 
$\Sigma$ such that $\+V$ is closed under
\begin{itemize}
	\itemAP ""Boolean combinations""\emph{:} if $\+R, \+S \in \+V_{\Sigma}$, then
		$\negrel\+R$, $\+R \cup \+S$ and $\+R \cap \+S$ belong to $\+V_{\Sigma}$ too,
	\itemAP ""Syntactic derivatives""\emph{:} if $\+R \in \+V_{\Sigma}$, then any "relation"
	"recognized@@sync" by the "syntactic synchronous algebra morphism" of $\+R$ also belongs
	to $\+V_{\Sigma}$.
	\itemAP ""Inverse morphisms""\emph{:} if $\phi\colon \Sync\Gamma \to \Sync\Sigma$ is
		a "synchronous algebra morphism" and $\+R \in \+V_{\Sigma}$ then
		$\phi^{-1}[\+R] \in \+V_{\Gamma}$. 
\end{itemize}

To recover a more traditional definition (of the form ``closure under Boolean operations, residuals\footnote{Also called ``quotient'' "eg" in \cite[\S III.1.3, p.~39]{pin_mathematical_2022}, or ``polynomial derivative'' in \cite[\S 4, p.~19]{bojanczyk_recognisable_2015}.} and inverse morphisms''), we need to properly define what are the residuals of a "relation". It 
turns out that the answer is quite surprising and less trivial than what one would expect.

\begin{definition}[Residuals]
	\label{def:residuals}
	Let $\?A$ be a "synchronous algebra", $\type{x}{\sigma} \in \?A$,
	and $C \subseteq \?A$ be a "closed subset".
	The \emph{left residual} and \emph{right residual} of $C$ by $\type{x}{\sigma}$ are defined by
	$\phantomintro{\residual}$
	\begin{align*}
		\reintro*\residual[\sigma]{x} C & \defeq
		\bigl\{
			\type{y}{\tau} \in \?A \mid
				\exists \type{y'}{\tau'} \congr{C} \type{y}{\tau},\;
				\type{x}{\sigma} \type{y'}{\tau'} \in C
		\bigr\}, \text{ and} \\
		C\reintro*\residual[\sigma]{x} & \defeq
		\bigl\{
			\type{y}{\tau} \in \?A \mid
				\exists \type{y'}{\tau'} \congr{C} \type{y}{\tau},\;
				\type{y'}{\tau'}\type{x}{\sigma} \in C
		\bigr\},
	\end{align*}
	respectively. We refer indiscriminately to both these notions as \AP""residuals"".
	We extend these notions to sets, by letting
	$\residual{X}{}C \defeq \bigcup_{x\in X}\residual{x}{}C$
	and $C\residual{X}{} \defeq \bigcup_{x\in X} C\residual{x}{}$.
\end{definition}

For the sake of readability, we will sometimes drop the "type" of elements when dealing
with "residuals".
It is routine to check that "residuals" are always "closed subsets" (since $\congr{C}$ is coarser than the "dependency relation"), or that $(\residual{x} C)\residual{y} =
\residual{x} (C\residual{y})$.
Equivalently, $C\residual[\sigma]{x}$ can be defined as the smallest "closed subset"
containing the ``naive residual''
$\bigl\{
	\type{y}{\tau} \in \?A \mid
		\type{y}{\tau}\type{x}{\sigma} \in C
\bigr\}$.
This latter set is always contained in $C\residual[\sigma]{x}$ (by reflexivity of $\congr{C}$),
and moreover, if it is empty, then so is $C\residual[\sigma]{x}$.

As an example, consider the relation $\+R$ from \Cref{ex:last_letter_is_a_if_big_diff}.
Then the ``naive right residual'' of $\proj{\+R}$ by $\type{\pair{a}{\pad}}{\lb}$
consists of $\type{\varepsilon}{\ll}$ and all elements of type $\lb$ and $\ll\to\lb$.
But it does not contain any element of type $\bl$ or $\ll\to\bl$ because such elements cannot be concatenated with $\type{\pair{a}{\pad}}{\lb}$ on the right.
Yet, the "residual" 
$\proj{\+R} \residual[\lb]{\pair{a}{\pad}}$ contains all elements of type $\bl$ (and also $\ll\to\bl$): for instance, $\type{\pair{\pad}{a}}{\bl} \in
\proj{\+R} \residual[\lb]{\pair{a}{\pad}}$ since $\type{\pair{\pad}{a}}{\bl} \congr{\+R} \type{\pair{a}{\pad}}{\lb}$
and $\type{\pair{a}{\pad}}{\lb} \type{\pair{a}{\pad}}{\lb} \in \+R$.

On the other hand, in the "algebra@@sync" $\Sync{a}$ consider the relation
$\+S = (aa)^*\times a(aa)^*$.
Then $\proj{\+S}\residual[\ll]{\pair{a}{a}}$ is empty since its 
``naive residual'' $\{\type{y}{\tau} \in \Sync{a} \mid \type{y}{\tau}\cdot\pair{a}{a} \in \+S\}$
is empty. Indeed, for $\type{y}{\tau}\cdot\type{\pair{a}{a}}{\ll}$ to
be well-defined, one needs $\tau$ to be $\ll$, "ie" $y$ encodes a pair of
two words $(u,v)$ of the same length. But then $(ua, va) \not\in \+S$.

\begin{restatable}{lemma}{lemmaCharacterizationPseudovar}
	\label{lemma:characterization-pseudovarieties-syncrel}
	A class $\+V\colon \Sigma \mapsto \+V_\Sigma$ is a "$\ast$-pseudovariety of synchronous relations" if, and only if, it is closed under "Boolean combinations", "residuals" and
	"inverse morphisms".
\end{restatable}
See the proof in \Cref{apdx-proof:lemma:characterization-pseudovarieties-syncrel}.

Let \AP$\B{V} \intro*\corrAR \+V$ denote the map (called \emph{correspondence}) that takes a 
"pseudovariety of synchronous algebras" and maps it to
\[\+V\colon \Sigma \mapsto \{\+R \subseteq \Sigma^*\times \Sigma^* \mid \SyntSA{\+R} \in \B{V}\}.\]

Dually, let \AP$\+V \intro*\corrRA \B{V}$ denote the \emph{correspondence} that takes
a "$\ast$-pseudo\-variety of synchronous relations" $\+V$
and maps it to the "pseudovariety of synchronous algebras" "generated@@var" by
all $\SyntSA{\+R}$ for some $\+R \in \+V_{\Sigma}$.
Here, the ""pseudovariety generated"" by a class $C$
of finite "locally generated" "synchronous algebras"
is the smallest "pseudovariety@@syncalg" containing
all finite "locally generated" "algebras@@sync" of $C$,
or equivalently,\footnote{The proof is straightforward,
see "eg" \cite[Proposition XI.1.1, p.~190]{pin_mathematical_2022} for a proof in the context of semigroups.} the class of all finite "locally generated" "synchronous algebras" 
that "divide@@sync" a finite product of "algebras@@sync" of $C$.\footnote{Note that ``being "locally generated"'' is not preserved by taking "subalgebras", but this is not an issue: we restrict the construction to (finite) "locally generated" "algebras@@sync".}

\AP\phantomintro(sync){Eilenberg correspondence theorem}\vspace{-1em}
\begin{restatable}[\reintro{An Eilenberg theorem for synchronous relations}]{lemma}{thmeilenberg}
	\label{lem:eilenberg-sy}
	The correspondences $\B{V} \corrAR \+V$ and $\+V \corrRA \B{V}$ define
	mutually inverse bijections between "pseudovarieties of
	synchronous algebras" and "$\ast$-pseudovarieties of synchronous relations".
\end{restatable}
See the proof in \Cref{apdx-proof:lem:eilenberg-sy}.

As consequence of \Cref{lem:eilenberg-sy}, if
$\+V$ is a "$\ast$-pseudovariety of synchronous relations"
and $\B{V}$ is a "pseudovariety of synchronous algebras",
we write \AP$\+V \intro*\corr \B{V}$
to mean that either $\+V \corrRA \B{V}$ or, equivalently, $\B{V} \corrAR \+V$.
This relation is called an \AP""Eilenberg-Schützenberger correspondence@@sync"".

\begin{proposition}
	If~~$\B{V}$ is a "pseudovariety of monoids", then \AP$\phantomintro{\projA}$
	\begin{align*}
		\reintro*\projA{\B{V}} & \defeq
		\{\?A \text{ "locally generated@@sync" finite "synchronous algebra"} \\
		& \qquad\qquad \text{ "st" all "underlying monoids" of $\?A$ are in $\B{V}$}\}
	\end{align*}
	is a "pseudovariety of synchronous algebras". Similarly,
	if~~$\+V$ is an "$\ast$-pseudovariety of regular languages", then
	the class of "$\+V$-relations", namely
	\[
		\intro*\projL{\+V} \colon
		\Sigma \mapsto \{\+R \subseteq \Sigma^* \times \Sigma^* \mid \exists L \in \+V_{\SigmaPair},\, \proj{\+R} = L \cap \WellFormed \},
	\]
	is a "$\ast$-pseudovariety of synchronous relations".
\end{proposition}

\begin{proof}
	The first point is straightforward. The second one follows from it and \Cref{lem:eilenberg-sy,thm:lifting-theorem-monoids}.
\end{proof}

\AP\phantomintro{Lifting Theorem: Pseudovariety Formulation}
Finally, \Cref{thm:lifting-theorem-monoids} can be elegantly rephrased
by saying that correspondences between "pseudovarieties of monoids"
and "$\ast$-pseudovarieties of regular languages" lift to correspondences
between "pseudovarieties of synchronous algebras" and
"$\ast$-pseudovarieties of synchronous relations".

\begin{theorem}[\reintro{Lifting Theorem: Pseudovariety Formulation}]
	\label{thm:lifting-theorem-monoids-pseudovarieties}
	If, in the "Eilenberg correspondence@@lang"
	between "pseudovarieties of mon\-oids" and "$\ast$-pseudovarieties of regular languages"
	we have~$\+V \corr \B{V}$,
	then in the "Eilenberg correspondence@@sync"
	between the "pseudovariety of synchronous algebras" $\projA{\B{V}}$ and
	the "$\ast$-pseudovariety of synchronous relations" $\projL{\+V}$,
	we have~$\projL{\+V} \corr \projA{\B{V}}$.
\end{theorem}
\section{Discussion}
\label{sec:discussion}

A natural next step is to generalize \Cref{quest:V-relations} by replacing
$\WellFormed$ by a fixed regular language $\Omega$.

\begin{question}
	Given a class of regular languages $\+V$, can we characterize (and decide) all
	languages of the form $L \cap \Omega$ for some $L \in \+V$?
\end{question}

We claim that the construction of "synchronous algebras"
can be generalized for any $\Omega$, giving rise to the notion of ``path algebras''.
The "lifting theorem for monoids" can be shown to hold for some $\Omega$, including "well-formed words" for $n$-ary relations with $n\geq 3$, and that it cannot effectively hold for all $\Omega$.

A natural next step would then be to study the relationship between ``path algebras'' and
Figueira \& Libkin's $L$-controlled relations \cite[\S 3]{figueira_synchronizing_2015},
see also \cite{descotte_resynchronizing_2018}.

Lastly, it would be interesting to extend the results on algebras to automata: for instance, can
we adapt our proof to show the existence of a minimal \emph{synchronous} automaton for each relation?

\clearpage
\appendix
\section{Monads Everywhere!}
\label{apdx:monads}

We denote by \AP$\intro*\Set[\+S]$ and $\intro*\Pos[\+S]$ the categories of $\+S$-typed sets
and $\+S$-partially ordered sets---note that in this model,
each type is equipped with its own order and
that elements of different types cannot be compared.
Similarly, let \AP$\intro*\Dep[\+S]$ be the category of $\+S$-"dependent sets".

We claim that "synchronous algebras" correspond to "Eilenberg-Moore algebras" of some monad over the category $\Dep[\types]$. For the sake of readability, we represent the underlying
"typed set" of a 
$\types$-"dependent set"
\[\?X = \langle X_{\ll},\, X_{\ll\to\lb},\, X_{\lb},\, X_{\ll\to\bl},\, X_{\bl} \rangle\]
as follows:
\PictureTypedSet{X_{\ll}}{X_{\ll\to\lb}}{X_{\lb}}{X_{\ll\to\bl}}{X_{\bl}.}

We define the \emph{synchronous monad}
\AP$\intro*\MonadSync$ over $\Dep[\types]$ as the functor which maps 
\PictureTypedSet{A}{B}{C}{D}{E}
equipped with a dependency relation $\dep$ to the "dependent set"
\PictureTypedSet{A^*}{A^*BC^* \cup A^*C^*}{C^*}{A^*DE^* \cup A^*E^*}{E^*,}
and two words are "dependent" if their domain are isomorphic
and their letters are pairwise dependent.
The unit and free multiplication are defined naturally.

Note in particular that all five empty words are "mutually dependent",
and that "synchronous words" $\Sync\Sigma$ correspond to applying
$\MonadSync$ to
\PictureTypedSet{\Sigma\times\Sigma}{\emptyset}{\Sigma\times\{\pad\}}{\emptyset}{\{\pad\}\times\Sigma,}
equipped with equality.
Moreover, "synchronous algebras" exactly correspond to $\MonadSync$-algebras.


A systemic approach to algebraic language theory was proposed by Bojańczyk using the formalism of monads \cite{bojanczyk_recognisable_2015}, for monads over finitely typed sets $\Set$.
Urbat, Adámek, Chen \& Milius then extended these results to capture monads over varieties of typed (ordered) algebras \cite{urbat_eilenberg_2017}.
Lastly, Blumensath extended those results to monads over the category of typed posets $\Pos$
when the set $\+S$ of types is infinite \cite{blumensath_algebraic_2021}. 

Observe that the category of "dependent sets" is not captured by any of the results above since
the "dependency relation" can compare elements of different types, contrary to typed posets \& co. 
\section{Characterizing Induced Classes of Recognizable Relations}
\label{apdx:relations}
\label{apdx:relations-recognizable}

In this section, we prove a simple characterization of 
"recognizable relations" whose languages are ``simple''.

\begin{proposition}
	Let $\mathbb{V}$ be a pseudovariety of finite monoids
	and $\mathcal{V}$ be the corresponding pseudovariety of regular languages.
	Given a relation $\mathcal{R} \subseteq \Sigma^* \times \Sigma^*$, the following
	are equivalent:
	\begin{enumerate}
		\item $\mathcal{R}$ is recognized by a monoid in $\B{V}$, \textit{i.e.}
			there is a monoid morphism $\phi\colon \Sigma^*\times\Sigma^* \to M$
			for some monoid $M \in \mathbb{V}$ such that $\mathcal{R} = \phi^{-1}[\textrm{Acc}]$, and
		\item There exists $n \in \mathbb{N}$ and $K_1,\hdots,K_n,L_1,\hdots,L_n$
			in $\mathcal{V}$ such that $\mathcal{R} = \bigcup_{i=1}^n K_i \times L_i$,
	\end{enumerate}
	in which case we say that $\mathcal{R}$ is $\mathcal{V}$-recognizable.
\end{proposition}

The statement of the result when $\B{V}$ is the pseudovariety of all regular languages is folklore,
see "eg" \cite[Corollary II.2.20, p.~254]{sakarovitch2009elements},
and the proof of the proposition is easy.

\begin{proof}
	\proofcase{$(1) \Rightarrow (2)$.} 
	By definition:
	\[
		\+R = \bigcup_{z \in \Acc} \phi^{-1}[z].
	\]
	Observe then that $\phi(u,v) = \phi((u,\varepsilon)\cdot(\varepsilon,v)) = \phi((u,\varepsilon))\cdot \phi((\varepsilon,v))$ for all $u,v \in \Sigma^* \times \Sigma^*$, and hence:
	\[
		\+R = \bigcup_{\substack{x,y \in M\\ \text{"st" } x\cdot y \in \Acc}}
		\underbrace{\{u \in \Sigma^* \mid \phi(u, \varepsilon) = x\}}_{\defeq \;K_x}
		\times \underbrace{\{v \in \Sigma^* \mid \phi(\varepsilon, v) = y\}}_{\defeq \;L_y}.
	\]
	Since $M$ is finite, the union is finite, and moreover, each $K_x$ and $L_y$ is
	recognized by $M$, and hence belong to $\+V$.

	\proofcase{$(2) \Rightarrow (1)$.} If $\mathcal{R} = \bigcup_{i=1}^n K_i \times L_i$ 
		where all languages belong to $\+V$, then let $M_i,N_i \in \B{V}$
		be their syntactic monoid, $\zeta_i,\eta_i$ be their
		syntactic morphism, and $\Acc_i,\Bcc_i$ be their accepting sets. Consider the monoid morphism
		\begin{center}
		\begin{tabular}{ccc}
			$\Sigma^* \times \Sigma^*$ & $\to$ & $\prod_{i}(M_i \times N_i)$ \\
			$(u,v)$ & $\mapsto$ & $\bigl( \zeta_i(u), \eta_i(v) \bigr)_{i}$.
		\end{tabular}
		\end{center}
		Then $\+R$ is the preimage by this morphism of
		\[
			\bigcup_{i=1}^n \bigl(
			\hdots \times (M_{i-1}\times N_{i-1}) \times
			(\Acc_i \times \Bcc_i) \times (M_{i+1}\times N_{i+1}) \times \hdots \bigr).
		\]
		The conclusion follows from the fact that $\B{V}$ is closed under (finite) products.
\end{proof}
\section{Pointers to Notions of Algebraic Language Theory}
\label{apdx:pointers-pin}

\begin{itemize}
	\itemAP ""idempotent"" elements:
		see \cite[\S II.1.2, p.~14]{pin_mathematical_2022};
	\itemAP semigroup division and ""monoid division"":
		see \cite[\S II.3.3, p.~21]{pin_mathematical_2022};
	\itemAP ""pseudovarieties of semigroups"" and ""pseudovarieties of monoids"":
		see \cite[\S XI.1, p.~189]{pin_mathematical_2022} under the name ``variety'';
	\itemAP ""$\ast $-pseudovarieties of regular languages""
		see \cite[\S XIII.3, p.~226]{pin_mathematical_2022};
	\itemAP ""$+$-pseudovarieties of regular languages"":
		see \cite[\S XIII.4, ``Eilenberg’s $+$-varieties'', p.~229]{pin_mathematical_2022};
	\itemAP ""Eilenberg correspondence@@lang"":
		see \cite[Theorem XIII.4.10, p.~228]{pin_mathematical_2022};
	\itemAP profinite topology and the $\intro*\profSg{-}$ notation:
		see \cite[\S X.2, p.~178]{pin_mathematical_2022};
	\itemAP ""profinite equality"" and the $\intro*\profeq$ notation:
		see \cite[\S XI.3, p.~193]{pin_mathematical_2022} under
		the name ``profinite identity for semigroups/monoids'';
	\itemAP ""profinite implication"" and the $\intro*\profimp$ notation:
		see \cite[\S XIII.1, p.~223]{pin_mathematical_2022};
	\itemAP ""satisfiability@@profeq"" of profinite equality:
		see \cite[\S XI.3, p.~193]{pin_mathematical_2022};
	\itemAP equations ""defining@@profeq"" a pseudovariety:
		see \cite[\S XI.3.3, p.~194]{pin_mathematical_2022}.
\end{itemize}
\section{Details on Synchronous Algebras}

\subsection{Proof of the Syntactic Morphism Theorem}
\label{apdx-proof:lem:syntactic-morphism-theorem}

\syntacticmorphismtheorem*

\begin{proof}[Proof of \Cref{lem:syntactic-morphism-theorem}.]
	\proofcase{Uniqueness.} Consider two potential "syntactic morphisms@@sync",
	say $\eta_1\colon \Sync\Sigma \surj \?A_1$ and $\eta_2\colon \Sync\Sigma \surj \?A_2$.
	Then by the universal property of $\eta_1$ (resp.~$\eta_2$), there exists
	$\psi_1\colon \?A_2 \surj \?A_1$ and $\psi_2\colon \?A_1 \surj \?A_2$ such that
	$\eta_1 = \psi_1 \circ \eta_2$ and $\eta_2 = \psi_2 \circ \eta_1$. Overall, it implies that
	the following digram commutes
	\begin{center}
		\begin{tikzcd}
			& \?A_1 \\
			\Sync\Sigma
				\ar[ur, "\eta_1", twoheadrightarrow]
				\ar[r, "\eta_2" description, twoheadrightarrow]
				\ar[dr, "\eta_1" swap, twoheadrightarrow]
			& \?A_2 \ar[u, "\psi_1" swap, twoheadrightarrow] \\
			& \?A_1, \ar[u, "\psi_2" swap, twoheadrightarrow]
		\end{tikzcd}
	\end{center}
	and so $\psi_1\circ \psi_2 \circ \eta_1 = \eta_1$. Since $\eta_1$
	is surjective, and hence right-cancellative, $\psi_1\circ \psi_2 = \id_{\?A_1}$.
	Symmetrically, $\psi_2 \circ \psi_1 = \id_{\?A_2}$, showing that
	$\psi_1$ and $\psi_2$ are mutually inverse "isomorphisms" of "synchronous algebras".

	\proofcase{Existence.} It follows from \Cref{prop:quotient-by-syntactic-congruence-is-syntactic-algebra}---which we will prove
	just afterwards---applied to the identity morphism $\id{}\colon \Sync\Sigma \surj \Sync\Sigma$,
	which "recognizes@@sync" $\+R$ since $\proj{\+R} = \id^{-1}[\proj{\+R}]$ and $\proj{\+R}$
	is "closed".
\end{proof}

\propquotientbysyntacticcongruenceissyntacticalgebra*

\begin{proof}
	We claim that $\quotient{\phi}{\congr{\Acc}}$
	is \emph{a}, and hence \emph{the},
	"syntactic synchronous algebra morphism" of $\+R$.
	First, if $\phi(u) \congr{\Acc} \phi(v)$ 
	then in particular $\phi(u) \in \Acc$ "iff" $\phi(v) \in \Acc$.
	It follows that the preimage of
	$\Acc' \defeq \equivclass{\Acc}{\congr{\Acc}}$ by $\quotient{\phi}{\congr{\Acc}}$
	is indeed $\proj{\+R}$. Moreover $\Acc'$ is a "closed subset" of
	$\quotient{\?A}{\congr{\Acc}}$ since $\Acc$ is a "closed subset" of $\?A$.
	Hence, $\quotient{\phi}{\congr{\Acc}}$ is a surjective "morphism@@sync"
	which "recognizes@@sync" $\+R$.

	Then, given another surjective "morphism@@sync" $\psi\colon \Sync\Sigma \surj \?B$
	which "recognizes@@sync" $\+R$, say $\proj{\+R} = \psi^{-1}[\Bcc]$, then for
	each $\type{b}{\sigma} \in \?B$, there exists $\type{u}{\sigma} \in \Sync\Sigma$ such that
	$\psi(\type{u}{\sigma}) = \type{b}{\sigma}$. This defines a map
	$\chi\colon \?B \to \quotient{\?A}{\congr{\Acc}}$ which
	sends $\type{b}{\sigma}$ to $\quotient{\phi}{\congr{\Acc}}(\type{u}{\sigma})$.

	We claim that $\quotient{\phi}{\congr{\Acc}} = \chi \circ \psi$,
	meaning that the following diagram commutes
	\begin{center}
		\begin{tikzcd}[ampersand replacement=\&]
			\& \?A \arrow{d}{\quotient{-}{\congr{\Acc}}} \\
			\Sync{\Sigma} 
				\arrow[twoheadrightarrow, swap]{ur}{\phi}
				\arrow[twoheadrightarrow]{r}{\quotient{\phi}{\congr{\Acc}}}
				\arrow[twoheadrightarrow, swap]{dr}{\psi}
			\& \SyntSA{\+R} \\
			\& \?B. \arrow[swap]{u}{\chi} \\
		\end{tikzcd}
	\end{center}
	Indeed, for
	$\type{u}{\sigma} \in \Sync\Sigma$, $\chi(\psi(\type{u}{\sigma}))$ equals
	$\quotient{\phi}{\congr{\Acc}}(\type{v}{\sigma})$ for some $\type{v}{\sigma} \in \Sync\Sigma$ 
	such that $\psi(\type{v}{\sigma}) = \psi(\type{u}{\sigma})$.
	This in turns implies that $\type{u}{\sigma} \congr{\+R} \type{v}{\sigma}$ since
	for all $x, y \in \Sync\Sigma$, if $xuy$ and $xvy$ are well-defined, 
	then so are $\psi(x)\psi(u)\psi(v)$ and $\psi(x)\psi(v)\psi(y)$,
	and both elements are equal, so one belongs to $\Bcc$ "iff" the other does.
	It follows that $xuy \in \+R$ "iff" $xvy \in \+R$, and hence
	$\type{u}{\sigma} \congr{\+R} \type{v}{\sigma}$.
	By surjectivity of $\phi$, it follows that
	$\phi(\type{u}{\sigma}) \congr{\Acc} \phi(\type{v}{\sigma})$,
	"ie" $\quotient{\phi}{\congr{\Acc}}(\type{u}{\sigma}) = \quotient{\phi}{\congr{\Acc}}(\type{v}{\sigma})$.
	And hence $\chi(\psi(u)) = \quotient{\phi}{\congr{\Acc}}(\type{u}{\sigma})$.

	We now show that $\chi$ is a "morphism@@sync".
	From $\quotient{\phi}{\congr{\Acc}} = \chi \circ \psi$ it follows that
	the map $\chi$ preserves the "product"\footnote{See "eg" \cite[Lemma 3.2, p.~10]{bojanczyk_recognisable_2015}.} and is surjective.
	Lastly, we claim that it preserves the "dependency relation". Indeed,
	by surjectivity of $\psi$, it boils down to proving that
	for all $\type{u}{\sigma}, \type{v}{\tau} \in \Sync\Sigma$, if $\psi(u) \dep \psi(v)$ then
	$\chi(\psi(u)) \dep \chi(\psi(v))$, namely $\quotient{\phi}{\congr{\Acc}}(\type{u}{\sigma}) \dep
	\quotient{\phi}{\congr{\Acc}}(\type{v}{\tau})$, which
	rewrites as $\phi(\type{u}{\sigma}) \congr{\Acc} \phi(\type{v}{\tau})$.
	So pick $x, y$ such that both $xuy$ and $xvy$ are well-defined. Then
	$\psi(x)\psi(u)\psi(y)$ and $\psi(u)\psi(v)\psi(y)$ have the same "type"
	as $xuy$ and $xvy$, respectively, so they are well-defined,
	but since $\psi(u) \dep \psi(v)$, then $\psi(x)\psi(u)\psi(y) \dep \psi(x)\psi(u)\psi(y)$
	and since $\Bcc$ is a "closed subset", one of these elements belongs in
	it "iff" the other ones does too, from which it follows
	that $xuv \in \+R$ "iff" $xvy \in \+R$ "ie" $\type{u}{\sigma} \congr{\+R} \type{v}{\tau}$
	and hence $\phi(\type{u}{\sigma}) \congr{\Acc} \phi(\type{v}{\tau})$.
	Overall, this proves that $\chi$ is a surjective "synchronous algebra morphism",
	and concludes our proof.
\end{proof}

\subsection{A Syntactic Synchronous Algebra}
\label{apdx-ex:last_letter_is_a_if_big_diff}

\begin{example}
	\label{ex:last_letter_is_a_if_big_diff}
	We now provide an example of "syntactic synchronous algebra" whose "dependency relation"
	is not an equivalence relation.
	Consider the relation
	\[
		\+R = \pair{\Sigma}{\Sigma}^*\bigl[
			\pair{\Sigma}{\pad} + \pair{\Sigma}{\pad}^+ \pair{a}{\pad}
			+ \pair{\pad}{\Sigma} + \pair{\pad}{\Sigma}^+ \pair{\pad}{a}
		\bigr],
	\]
	where $\Sigma = \{a,b\}$.
	In other words, a pair $(u,v)$ belongs to $\+R$ if either:
	the length difference between $u$ and $v$ is one, or 
	it is at least two and the longer words ends with an `$a$'.
	We are going to compute the "syntactic congruence" $\congr{\+R}$ of $\proj{\+R}$ in $\Sync\Sigma$.

	\begin{figure}[htbp]
		\begin{center}
			\includegraphics[width=.6\linewidth]{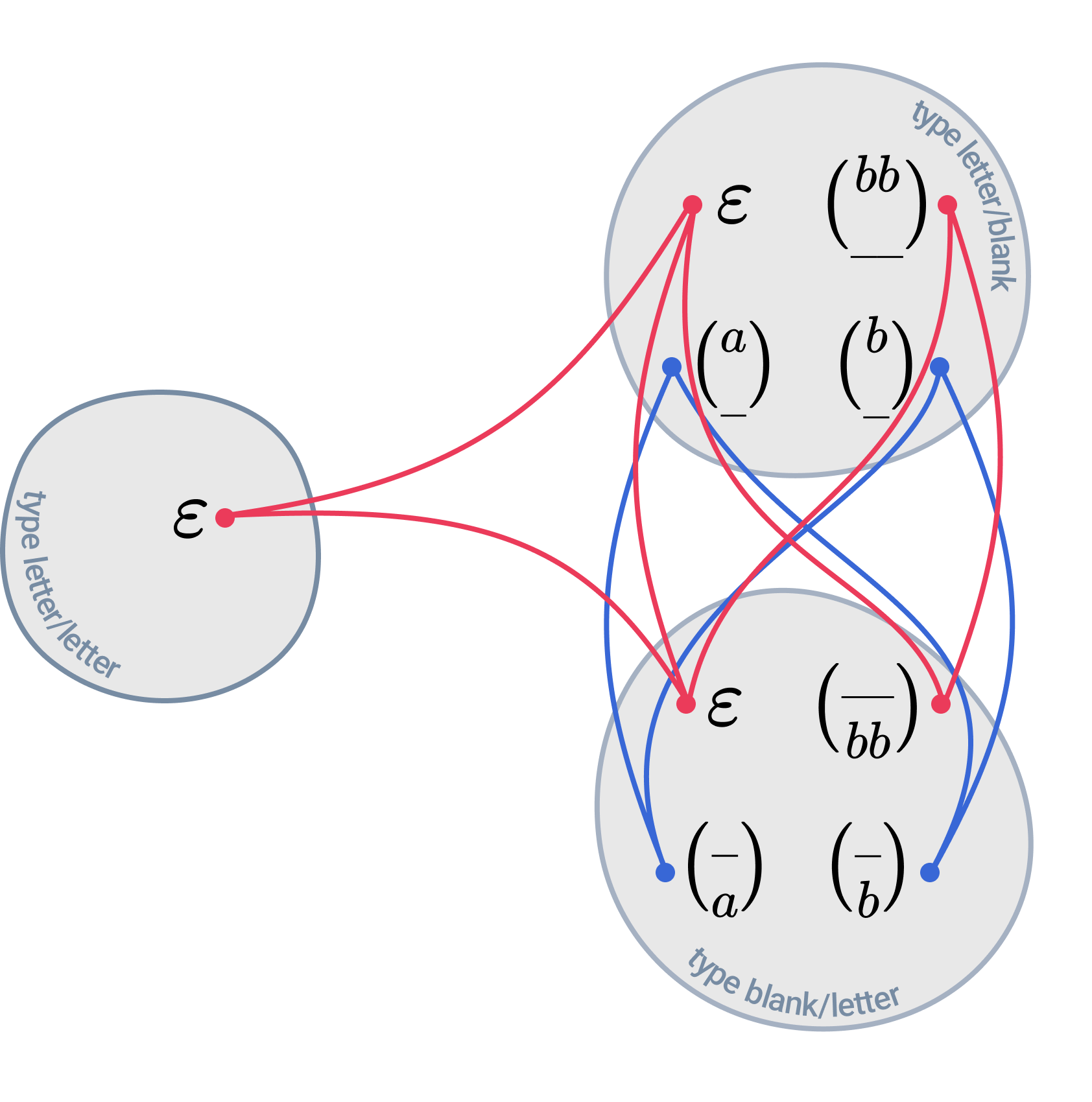}
		\end{center}
		\caption{
			\label{fig:example_syntactic_algebra}
			The three "underlying monoids" of the "syntactic synchronous algebra" of the relation
			$\+R$ of \Cref{ex:last_letter_is_a_if_big_diff}, together with its "dependency relation". The two equivalence classes of its transitive closure are drawn in red and blue.
		} 
	\end{figure}

	Recall that the restriction of $\congr{\+R}$ to a single "type" is
	an equivalence relation. For $\ll$, there is a single equivalence class,
	denoted by/identified with $\type{\varepsilon}{\ll}$.
	For "type" $\lb$, we claim that for any $u \in \Sigma^*$:
	\begin{multicols}{2}
	\begin{itemize}
		\item $\type{\pair{u}{\pad}}{\lb} \congr{\+R} \type{\pair{a}{\pad}}{\lb}$ "iff" $u \in \Sigma^* a$;
		\item $\type{\pair{u}{\pad}}{\lb} \congr{\+R} \type{\pair{b}{\pad}}{\lb}$ "iff" $u = b$;
		\item $\type{\varepsilon}{\lb}$ is alone in its equivalence class;
		\columnbreak
		\item all elements of the form $\type{\pair{u}{\pad}}{\lb}$ where $u$ is a word of length at least 2 whose last letter is a `$b$' are pairwise equivalent, and the class is identified with $\type{\pair{bb}{\pad}}{\lb}$.
	\end{itemize}
	\end{multicols}
	Since $\+R$ is invariant under $(u,v) \mapsto (v,u)$, the situation is symmetric for "type" $\bl$. Moreover, $\+R$ is invariant under adding/removing prefixes of type $\ll$,
	so "types" $\ll \to \lb$ and $\ll\to\bl$ also have four elements each.

	In the end, we obtain the "synchronous algebra" drawn in \Cref{fig:example_syntactic_algebra} (elements of "type" $\ll\to\lb$ and $\ll\to\bl$ are omitted for the sake of simplicity).
	Note that the "dependency relation" is not transitive:
	for instance $\type{\pair{a}{\pad}}{\lb} \congr{\+R} \type{\pair{\pad}{b}}{\bl}$
	and $\type{\pair{\pad}{b}}{\bl} \congr{\+R} \type{\pair{b}{\pad}}{\lb}$
	but $\type{\pair{a}{\pad}}{\lb} \not\congr{\+R} \type{\pair{b}{\pad}}{\lb}$.

Given two elements of distinct "type", we want to determine
when they are "dependent". For the sake of simplicity,
we will focus on types $\ll, \lb$ and $\bl$.

\proofcase{Types $\ll$ and $\lb$.}
Since "dependent" elements must either both belong to $\+R$ or both to $\negrel\+R$,
we have $\type{\varepsilon}{\ll} \not\congr{\+R} \type{\pair{a}{\pad}}{\lb}$
and $\type{\varepsilon}{\ll} \not\congr{\+R} \type{\pair{b}{\pad}}{\lb}$.
Because $\proj{\+R}$ is a "closed subset",
then $\type{\varepsilon}{\ll} \congr{\+R} \type{\varepsilon}{\lb}$.
Moreover, $\type{\varepsilon}{\ll} \not\congr{\+R} \type{\pair{bb}{\pad}}{\lb}$ since
$\type{\varepsilon}{\ll} \type{\pair{b}{\pad}}{\lb} \in \+R$
but $\type{\pair{bb}{\pad}}{\lb} \type{\pair{b}{\pad}}{\lb} \not\in \+R$.

\proofcase{Types $\lb$ and $\bl$.} Again, using the fact that "dependent" elements must either both belong to $\+R$ or both to $\negrel\+R$, we have
\begin{align*}
	\type{\pair{a}{\pad}}{\lb} & \not\congr{\+R} \type{\varepsilon}{\bl},
	& \type{\pair{a}{\pad}}{\lb} & \not\congr{\+R} \type{\pair{\pad\pad}{bb}}{\bl},
	& \type{\pair{b}{\pad}}{\lb} & \not\congr{\+R} \type{\varepsilon}{\bl},
	& \type{\pair{b}{\pad}}{\lb} & \not\congr{\+R} \type{\pair{\pad\pad}{bb}}{\bl},\\
	\type{\varepsilon}{\lb} & \not\congr{\+R} \type{\pair{\pad}{a}}{\bl},
	& \type{\varepsilon}{\lb} & \not\congr{\+R} \type{\pair{\pad}{b}}{\bl},
	& \type{\pair{bb}{\pad\pad}}{\lb} & \not\congr{\+R} \type{\pair{\pad}{a}}{\bl},
	& \type{\pair{bb}{\pad\pad}}{\lb} & \not\congr{\+R} \type{\pair{\pad}{b}}{\bl}.
\end{align*}
We claim that $\type{\pair{a}{\pad}}{\lb} \congr{\+R} \type{\pair{\pad}{a}}{\bl}$. Note that there is no,
$\type{y}{\tau}$ be such that
$\type{\pair{a}{\pad}}{\lb}\type{y}{\tau}$
and $\type{\pair{\pad}{a}}{\bl}\type{y}{\tau}$ are "well-formed".
So, let $\type{x}{\sigma}$ be such that $\type{x}{\sigma}\type{\pair{a}{\pad}}{\lb}$
and $\type{x}{\sigma}\type{\pair{\pad}{a}}{\bl}$ are "well-formed":
$\sigma$ must be of type $\ll$, but then $\+R$ is invariant under removing prefixes of type $\ll$, and so $\type{x}{\sigma}\type{\pair{a}{\pad}}{\lb} \in \+R$ and 
$\type{x}{\sigma}\type{\pair{\pad}{a}}{\bl} \in \+R$. And hence,  $\type{\pair{a}{\pad}}{\lb} \congr{\+R} \type{\pair{\pad}{a}}{\bl}$. Similarly,
\[
	\type{\pair{a}{\pad}}{\lb} \congr{\+R} \type{\pair{\pad}{b}}{\bl},\quad 
	\type{\pair{b}{\pad}}{\lb} \congr{\+R} \type{\pair{\pad}{a}}{\bl}\quad\text{and}\quad
	\type{\pair{b}{\pad}}{\lb} \congr{\+R} \type{\pair{\pad}{b}}{\bl}.\quad 
\]
The dual argument (using now the fact that both sides do \emph{not} belong to $\+R$) shows that
\[
	\type{\varepsilon}{\lb} \congr{\+R} \type{\varepsilon}{\bl},\quad
	\type{\varepsilon}{\lb} \congr{\+R} \type{\pair{\pad\pad}{bb}}{\bl},\quad
	\type{\pair{bb}{\pad\pad}}{\lb} \congr{\+R} \type{\varepsilon}{\bl}\quad\text{and}\quad
	\type{\pair{bb}{\pad\pad}}{\lb} \congr{\+R} \type{\pair{\pad\pad}{bb}}{\bl}.
\]

The case of types $\ll$ and $\lb$ follows by symmetry.
\end{example}
\section{Details on Pseudovarieties}

\subsection{Proof of the Lifting Theorem}
\label{apdx-proof:thm:lifting-theorem-monoids}

\liftingtheoremmonoids*

\begin{proof}
	\proofcase{(1) $\Rightarrow$ (2).} Since $\+R$ is a "$\+V$-relation", there exists
	$\+L \in \+V_{\SigmaPair}$ such that $\proj{\+R} = \+L \cap \WellFormed$.
	Hence, there exists a morphism of monoids $\phi\colon (\SigmaPair)^* \to M$ such 
	that $M \in \B{V}$ and $\+L = \phi^{-1}[\Acc]$ for some $\Acc \subseteq M$.
	It follows that $\proj{\+R} = \+L \cap \WellFormed[\Sigma]$ rewrites as
	``for all $u \in \WellFormed$, $\phi(u) \in \Acc$ "iff" $u \in \proj{\+R}$''.
	Letting $\inducedalg{M}$ be the "synchronous algebra induced by the monoid" $M$, define
	$\psi\colon \Sync\Sigma \to \inducedalg{M}$ by $\psi(\type{u}{\sigma}) \defeq
	\type{(\phi(u))}{\sigma}$ for $\type{u}{\sigma} \in \Sync\Sigma$.
	Let $\Acc' \defeq \{\type{x}{\sigma} \mid x \in \Acc \land \sigma \in \types\}$.
	We claim that $\psi^{-1}[\Acc'] = \proj{\+R}$. Indeed, for $\type{u}{\sigma} \in \Sync\Sigma$,
	$\type{u}{\sigma} \in \proj{\+R}$ "iff" $u \in \+L$,
	"ie" $\phi(u) \in \Acc$,
	that is $\psi(\type{u}{\sigma}) = \type{(\phi(u))}{\sigma} \in \Acc'$.
	Notice then that all "underlying monoids" of $\inducedalg{M}$ are $M$,
	and hence they belong to $\B{V}$.

	\proofcase{(2) $\Rightarrow$ (3).} By \Cref{lem:syntactic-morphism-theorem}, 
	the "syntactic synchronous algebra" of $\+R$ "divides@@sync"
	any "algebra@@sync" $\?B$ "recognizing@@sync" $\+R$.
	In particular, its "underlying monoids" "divide@@monoid" the "underlying monoids" of $\?B$. The conclusion follows since $\B{V}$ is closed under "division@@monoid". 
	
	\proofcase{(3) $\Rightarrow$ (1).} Denote by $M_{\ll}, M_{\lb}$ and $M_{\bl}$ the 
	"underlying monoids" of $\SyntSA{\+R}$. Let $\Acc \subseteq \SyntSA{\+R}$ be the accepting set 
	such that $\proj{\+R} = \SyntSAM[-1]{\+R}[\Acc]$.
	Define $M \defeq M_{\ll} \times M_{\lb} \times M_{\bl}$, and
	\begin{center}
		\begin{tabular}{rccc}
			$\phi\colon$
			& $(\SigmaPair)^*$
			& $\to$
			& $M$\\
			& $\pair{a}{b}$
			& $\mapsto$
			& $\langle \SyntSAM{\+R}\type{\pair{a}{b}}{\ll},\, \type{1}{\lb},\, \type{1}{\bl} \rangle$ \\
			& $\pair{a}{\pad}$
			& $\mapsto$
			& $\langle \type{1}{\ll},\, \SyntSAM{\+R}\type{\pair{a}{\pad}}{\lb},\, \type{1}{\bl} \rangle$ \\
			& $\pair{\pad}{a}$
			& $\mapsto$
			& $\langle \type{1}{\ll},\, \type{1}{\lb},\, \SyntSAM{\+R}\type{\pair{\pad}{a}}{\bl} \rangle$
		\end{tabular}		
	\end{center}
	and finally, let
	\begin{align*}
		\Acc'
		& \cup \bigl\{
			\langle\type{1}{\ll},\, \type{1}{\lb},\, \type{z}{\bl}\rangle
			\mid \type{z}{\bl} \in \Acc
		\bigr\} \\
		& \cup \bigl\{
			\langle\type{x}{\ll},\, \type{y}{\lb},\, \type{1}{\bl}\rangle
			\mid \type{x}{\ll}\cdot \type{y}{\lb} \in \Acc
		\bigr\} \\
		& \cup \bigl\{
			\langle\type{x}{\ll},\, \type{1}{\lb},\, \type{z}{\bl}\rangle
			\mid \type{x}{\ll}\cdot\type{z}{\bl} \in \Acc
		\bigr\}.
	\end{align*}
	We first claim that
	\begin{align}
		\text{
			For every $\type{u}{\ll\to\lb} \in \Sync\Sigma$,
		} & \notag \\
		& 
		\hspace{-2cm}\text{
			$\phi(u)$ is of the form $\langle a,\, b,\, 1 \rangle$
			and moreover,
		}
		\label{eq:relationship-monoid-with-synchronous-algebra}
		\tag{\adforn{75}}\\
		& \hspace{-2cm} \text{
			$\SyntSAM{\+R}(\type{u}{\ll\to\lb}) = a \cdot b$,
		}\notag
	\end{align}
	which can trivially be proven by induction on $u$. Analogous results
	hold for words of different "type".
	We then prove that for each $\type{u}{\sigma} \in \Sync\Sigma$,
	\begin{equation}
		\SyntSAM{\+R}(\type{u}{\sigma}) \in \Acc
		\quad\text{"iff"}\quad
		\phi(u) \in \Acc'.
		\label{eq:accepting-set-is-preserved}
		\tag{\adforn{28}}
	\end{equation}
	The direct implication is straightforward, using
	\Cref{eq:relationship-monoid-with-synchronous-algebra}.
	The converse implication is more tricky: assume "eg" that $\sigma = \ll\to\lb$,
	say $\type{t}{\sigma} = \type{u}{\ll}\type{v}{\lb}$.
	If $\phi(t) \in \Acc'$, using \Cref{eq:relationship-monoid-with-synchronous-algebra} then it implies either that
	\begin{enumerate}
		\item $\SyntSAM{\+R}(\type{u}{\ll}) = \type{1}{\ll}$ and $\SyntSAM{\+R}(\type{v}{\lb}) = \type{1}{\ll}$, and $\type{1}{\bl} \in \Acc$, or
		\item $\SyntSAM{\+R}(\type{u}{\ll})\cdot \SyntSAM{\+R}(\type{v}{\lb}) \in \Acc$, or even
		\item $\SyntSAM{\+R}(\type{v}{\lb}) = \type{1}{\lb}$ and
			$\SyntSAM{\+R}(\type{u}{\ll})\cdot \type{1}{\bl} \in \Acc$.
	\end{enumerate}
	Clearly, $(2)$ implies the desired conclusion, namely that $\SyntSAM{\+R}(\type{t}{\sigma}) = \SyntSAM{\+R}(\type{u}{\ll})\SyntSAM{\+R}(\type{v}{\lb}) \in \Acc$.
	In all other cases, we will make heavy use of the "dependency relation".
	For case $(1)$, we have that $\SyntSAM{\+R}(\type{t}{\sigma}) = \type{1}{\ll\to\lb}$.
	From $\type{1}{\bl} \in \Acc$, \Cref{fact:closed-subset-units} yields $\type{1}{\ll} \cdot \type{1}{\lb} = \type{1}{\ll\to\lb} \in \Acc$, since $\Acc$ is "closed".
	Lastly, in case $(3)$, $\SyntSAM{\+R}(\type{u}{\ll}) \dep
	\SyntSAM{\+R}(\type{u}{\ll})\cdot \type{1}{\bl} \in \Acc$ so
	$\SyntSAM{\+R}(\type{u}{\ll}) \in \Acc$ and hence
	$\SyntSAM{\+R}(\type{t}{\ll\to\lb}) = \SyntSAM{\+R}(\type{u}{\ll}) \cdot \type{1}{\lb}
	\dep \SyntSAM{\+R}(\type{u}{\ll}) \in \Acc$ which yields
	$\SyntSAM{\+R}(\type{t}{\ll\to\lb}) \in \Acc$. This concludes the proof
	of \eqref{eq:accepting-set-is-preserved} for "type" $\sigma = \ll\to\lb$.
	Other "types" are handled similarly, and hence
	$\proj{\+R} = \phi^{-1}[\Acc']\cap\WellFormed$.
\end{proof}

\subsection{Proof of Lemma~\ref{lemma:characterization-pseudovarieties-syncrel}}
\label{apdx-proof:lemma:characterization-pseudovarieties-syncrel}

\begin{proposition}
	\label{prop:inverse-morphism-preserve-congruence}
	Let $\phi\colon \?A \surj \?B$ be a surjective "morphism@@sync",
	and $\Acc$ be a "closed subset" of $\?B$. Let $a, a' \in \?A$.
	Then
	\[
		a \congr{\phi^{-1}[\Acc]} a'
		\quad\text{"iff"}\quad
		\phi(a) \congr{\Acc} \phi(a').
	\]
\end{proposition}

\begin{proof}
	\proofcase{Direct implication.}
	Pick any $b_\ell, b_r \in \?B$ such that both
	$b_\ell \phi(a) b_r$ and
	$b_\ell \phi(a') b_r$
	are well-defined. By surjectivity of $\phi$, there exists
	$a_\ell, a_r \in \?A$
	such that $\phi(a_\ell) = b_\ell$
	and $\phi(a_r) = b_r$.
	Then both $a_\ell a a_r$
	and $a_\ell a' a_r$ are well-defined since they have the same "type"
	as $b_\ell \phi(a) b_r$ and $b_\ell \phi(a') b_r$, respectively.
	From $a \congr{\phi^{-1}[\Acc]} a'$,
	it follows that $a_\ell a a_r$ belongs to $\phi^{-1}[\Acc]$ "iff" $a_\ell a' a_r$ does.
	And hence
	\[
		b_\ell \phi(a) b_r \in \Acc 
	 	\quad\text{"iff"}\quad
		b_\ell \phi(a') b_r \in \Acc.
	\]

	\proofcase{Converse implication.}
	Dually, pick any $a_\ell, a_r \in \?A$ such that both
	$a_\ell a a_r$ and $a_\ell a' a_r$ are well-defined.
	Then $\phi(a_\ell) \phi(a) \phi(a_r)$ and $\phi(a_\ell) \phi(a') \phi(a_r)$
	are also well-defined since they have the same "type" as their preimage,
	and $\phi(a) \congr{\Acc} \phi(a')$ implies that the element $\phi(a_\ell) \phi(a) \phi(a_r)$ belongs
	to $\Acc$ "iff" $\phi(a_\ell) \phi(a') \phi(a_r)$ does. It follows
	that $a_\ell a a_r \in \phi^{-1}[\Acc]$ "iff" $a_\ell a' a_r \in \phi^{-1}[\Acc]$.
\end{proof}

\begin{proposition}[Inverse images of surjective "morphisms@@sync" preserve "residuals"]
	\label{prop:inverse-morphism-preserve-residuals}
	Let $\phi\colon \?A \surj \?B$ be a surjective "morphism@@sync", and $\Acc \subseteq \?B$
	be a "closed subset". Let $u \in \?A$. Then
	\[\residual{u}\phi^{-1}[\Acc] = \phi^{-1}[\residual{\phi(u)}\Acc].\]
\end{proposition}

\begin{proof}
	\proofcase{Left-to-right inclusion.}
	Let $a\in \residual{u}\phi^{-1}[\Acc]$.
	Then there exists $a' \in \?A$ such that $a \congr{\phi^{-1}[\Acc]} a'$
	and $ua' \in \phi^{-1}[\Acc]$.
	By \Cref{prop:inverse-morphism-preserve-congruence}
	$a \congr{\phi^{-1}[\Acc]} a'$ implies $\phi(a) \congr{\Acc} \phi(a')$,
	and $ua' \in \phi^{-1}[\Acc]$ yields $\phi(u)\phi(a') \in \Acc$.
	Overall, this shows that $a \in \phi^{-1}[\residual{\phi(u)}\Acc]$.

	\proofcase{Right-to-left inclusion.}
	Let $a \in \phi^{-1}[\residual{\phi(u)}\Acc]$. Then
	$\phi(a) \in \residual{\phi(u)}\Acc$, so there
	exists $b' \in \?B$ such that $\phi(a) \congr{\Acc} b'$
	and $\phi(u)b' \in \Acc$. By surjectivity of $\phi$ and
	\Cref{prop:inverse-morphism-preserve-congruence}, there exists $a' \in \?A$ 
	such that $\phi(a') = b'$ and $a \congr{\phi^{-1}[\Acc]} a'$.
\end{proof}

\lemmaCharacterizationPseudovar*

\begin{proof}[Proof of \Cref{lemma:characterization-pseudovarieties-syncrel}]
	\proofcase{Direct implication.} By \Cref{prop:inverse-morphism-preserve-residuals},
	the "residual" of any "relation" "recognized@@sync" by some "morphism@@sync"
	$\phi$ is also "recognized@@sync" by $\phi$. Hence, being closed under
	"syntactic derivatives" implies being closed under "residuals".

	\proofcase{Converse implication.} Consider some relation $\+R$.
	We will show that any relation "recognized@@sync" by $\SyntSAM{\+R}$
	can be expressed as a "Boolean combination" of "residuals" of $\+R$.\footnote{This result can be 
	put in perspective with \cite[Lemma XIII.4.11, p.~229]{pin_mathematical_2022} which proves a similar result in the context of monoids.}
	Let $\Acc$ be the "closed subset" of $\SyntSA{\+R}$ such that
	$\proj{\+R} = \SyntSAM{\+R}^{-1}[\Acc]$. Pick $x \in \SyntSA{\+R}$.
	Let $\Lambda \defeq \{s, t \in \SyntSA{\+R} \mid \exists x'\in\SyntSA{\+R},\; x' \dep x \text{ and } sx't \in \Acc\}$.
	We claim that
	\begin{equation}
		\equivclass{x}{\dep_{\SyntSA{\+R}}} =
		\left(
			\bigcap_{(s,t) \in \Lambda} \residual{s} \Acc\,\residual{t}
		\right)
		\smallsetminus
		\left(
			\bigcup_{(s, t) \not\in \Lambda} \residual{s} \Acc\,\residual{t}
		\right).
		\label{eq:boolean-combination-of-residuals}
		\tag{\adforn{41}}
	\end{equation}
	To prove the inclusion from left-to-right, first
	notice that $x \in \residual{s} \Acc\,\residual{t}$ for all $(s, t) \in \Lambda$.
	Then, assume by contradiction that there exists $(s, t) \not\in \Lambda$ "st"
	$x \in \residual{s} \Acc\,\residual{t}$. Then there would exist $x' \congr{\Acc} x$ such that
	such that $sx't \in \Acc$. But since $\SyntSAM{\+R}$ is the "syntactic synchronous algebra"
	of $\+R$, $\congr{\Acc}$ is precisely the relation $\dep$ by \Cref{coro:syntactic-congruence-is-syntactic-dependency}. Contradiction.
	Hence, $x$ belongs to the right-hand side (RHS). But then, this latter set is
	a Boolean combination of "residuals" of a "closed subset", so it is also
	"closed", and hence $\equivclass{x}{\dep_{\SyntSA{\+R}}}$ is included in the RHS.

	Dually, any element $y$ of the RHS satisfies that for all $s, t \in \SyntSA{\+R}$,
	$x \in \residual{s} \Acc\,\residual{t}$ "iff" $y \in \residual{s} \Acc\,\residual{t}$.
	We claim that $x \congr{\Acc} y$. Pick $s, t \in \?B$ and assume that
	both $sxt$ and $syt$ are well-defined. If $sxt \in \Acc$ then $x \in \residual{s} \Acc\,\residual{t}$ so $y \in \residual{s} \Acc\,\residual{t}$ and hence, there
	exists $y' \dep_{\SyntSA{\+R}} y$ "st" $sy't \in \Acc$. But $syt$ is also well-defined
	and $y \dep_{\SyntSA{\+R}} y'$ so $syt \in \Acc$. By symmetry, we have shown that
	$sxt \in \Acc$ "iff" $syt \in \Acc$, and hence $x \congr{\Acc} y$.
	Using again the fact that $\SyntSA{\+R}$ is the "syntactic algebra@@sync" of $\+R$, it 
	follows that $x \dep_{\SyntSA{\+R}} y$. This concludes the proof of
	\eqref{eq:boolean-combination-of-residuals}. By taking the union, it
	follows that any "closed subset" of $\SyntSA{\+R}$ is a Boolean combination
	of "residuals" of $\Acc$. Applying \Cref{prop:inverse-morphism-preserve-residuals}
	then yields that any "relation" "recognized@@sync" by $\phi$ is a Boolean combination of 
	"residuals" of $\+R$. Hence, any class closed under "Boolean combinations" and
	"residuals" is also closed under "syntactic derivatives". 
\end{proof}

\subsection{Proof of Theorem~\ref{lem:eilenberg-sy}}
\label{apdx-proof:lem:eilenberg-sy}

\thmeilenberg*

\begin{proof}
	We very roughly follow the proof schema of \cite[\S XIII.4, pp.~226--229]{pin_mathematical_2022}---which is
	a proof of Eilenberg's theorem in the context of monoids.

	\proofcase{The correspondence $\B{V} \corrAR \+V$ produces varieties.}
	First we have to show that if $\B{V}$ is a "pseudovariety
	of synchronous algebras" and $\B{V} \corrAR \+V$, then $\+V$ is a
	"$\ast$-pseudovarieties of synchronous relations".
	Since $\B{V}$ is closed under finite products, $\+V$ is closed under Boolean operations.

	\emph{"Syntactic derivatives":} Then let $\+R \in \+V_\Sigma$, and let $\+S$ be any other "relation"
	"recognized@@sync" by $\SyntSA{\+R}$. This implies that $\SyntSA{\+S}$
	"divides@@sync" $\SyntSA{\+R}$, and so $\SyntSA{\+S} \in \B{V}$, from which
	we have $\SyntSA{\+S} \in \+V_\Sigma$.

	\emph{"Inverse morphisms":} Lastly, if $\+R \in \+V_\Sigma$, say $\proj{\+R} = \SyntSAM[-1]{\+R}[\Acc]$,
	if $\psi\colon \Sync\Gamma \to \Sync\Sigma$ is a "synchronous algebra morphism",
	then $\psi^{-1}[\+R] = (\SyntSAM{\+R} \circ \psi)^{-1}[\Acc]$, so
	$\psi^{-1}[\+R]$ is "recognized@@sync" by $\SyntSA{\+R}$, that is
	$\SyntSA{\psi^{-1}[\+R]}$ "divides@@sync" $\SyntSA{\+R}$. Since $\SyntSA{\+R} \in \B{V}$
	and $\B{V}$ is closed by "division@@sync", it follows that $\SyntSA{\psi^{-1}[\+R]} \in \B{V}$
	and hence $\psi^{-1}[\+R] \in \+V_{\Gamma}$. 
	This concludes the proof that $\+V$ is a "$\ast$-pseudovariety of synchronous relations".

	\proofcase{Inverse bijections: part 1.} Assume that $\B{V} \corrAR \+V$ and $\+V \corrRA \B{W}$.
	Then
	\[\+V\colon \Sigma \mapsto \{\+R \subseteq \Sigma^*\times \Sigma^* \mid \SyntSA{\+R} \in \B{V}\},\]
	and so $\B{W}$ is the "pseudovariety@@syncalg" generated by all "syntactic synchronous algebras" that 
	belong to $\B{V}$.
	It follows that $\B{W} \subseteq \B{V}$.
	To prove that $\B{V} \subseteq \B{W}$, let $\?A \in \B{V}$.
	Let $\Sigma_{\?A}$ be an alphabet big enough so that
	there are injections from $\?A_{\ll}$ to $\Sigma_{\?A}\times \Sigma_{\?A}$,
	and from $\?A_{\lb}$ and $\?A_{\bl}$ to $\Sigma_{\?A} \times \pad$ and $\pad \times \Sigma_{\?A}$,
	respectively. Since $\?A$ is "locally generated@@sync",
	this allows us to define a surjective "synchronous algebra morphism"
	$\phi\colon \Sync{\Sigma_{\?A}} \surj \?A$.
	We then claim that $\?A$ "divides@@sync"
	$\?B \defeq \prod_{\type{x}{\tau} \in \?A} \?B_{\type{x}{\tau}}$
	where $\?B_{\type{x}{\tau}}$ is the "syntactic synchronous algebra"
	of $\phi^{-1}[\type{x}{\tau}]$.
	Indeed, let $\psi_{\type{x}{\tau}} \colon \Sync{\Sigma_{\?A}} \surj \?B_{\type{x}{\tau}}$
	be the "syntactic synchronous algebra morphism" of $\phi^{-1}[\type{x}{\tau}]$, say
	$\phi^{-1}[\type{x}{\tau}] = \psi_{\type{x}{\tau}}^{-1}[\Acc_{\type{x}{\tau}}]$.
	Then consider
	\[\begin{tabular}{rccc}
		$\Psi\colon$ & $\Sync{\Sigma_{\?A}}$ & $\to$ & $\?B$ \\
		& $\type{u}{\sigma}$ & $\mapsto$ &
			$\langle \psi_{\type{x}{\tau}}(\type{u}{\sigma}) \rangle_{\type{x}{\tau} \in \?A}$,
	\end{tabular}\]
	and let $\?B_0$ be its image. Observe that for each $\type{u}{\sigma} \in \Sync{\Sigma_{\?A}}$,
	$\psi_{\type{x}{\tau}}(\type{u}{\sigma}) \in \Acc_{\type{x}{\tau}}$ "iff"
	$\type{u}{\sigma} \in \phi^{-1}[\type{x}{\tau}]$ "ie" $\phi(\type{u}{\sigma}) = \type{x}{\tau}$---note by the way that it implies $\sigma = \tau$.
	This implies that for any $\type{(\langle y_{\type{x}{\tau}} \rangle_{\type{x}{\tau} \in \?A})}{\sigma} \in \?B_0$, there exists a unique $\type{x}{\tau}$ "st"
	$y_{\type{x}{\tau}} \in \Acc_{\type{x}{\tau}}$. This defines a map
	$\chi\colon \?B_0 \to \?A$. Since moreover it makes the following diagram commute
	
	\begin{center}\begin{tikzcd}
		\Sync{\Sigma_{\?A}} \ar[r, twoheadrightarrow, "\Psi"]
			\ar[dr, twoheadrightarrow, "\phi" swap]
		& \?B_0 \dar["\chi"] \\
		& \?A
	\end{tikzcd}\end{center}
	it follows that $\chi$ is in fact a surjective "synchronous algebra morphism".\footnote{See "eg" \cite[Lemma 3.2, p.~10]{bojanczyk_recognisable_2015} for a proof in a similar (but different) context.} Hence, $\?A$ is a "quotient@@sync" of $\?B_0$, which is a "subalgebra@@sync" of
	$\?B$, which in turns in a product of "algebras@@sync" from $\B{W}$, and so $\?A \in \B{W}$.
	It concludes the proof that $\B{V} = \B{W}$.

	\proofcase{Inverse bijections: part 2.} Assume now that $\+V \corrRA \B{V}$
	and $\B{V} \corrAR \+W$. Then for each $\Sigma$, for each $\+R \in \+V_{\Sigma}$,
	$\SyntSA{\+R} \in \B{V}$ so $\+R \in \+W_{\Sigma}$, and hence $\+V \subseteq \+W$.

	We then want to show the converse inclusion, namely $\+W \subseteq \+V$. Let $\+R \in \+W_{\Sigma}$ for some $\Sigma$, "ie" $\SyntSA{\+R} \in \B{V}$.
	Hence there exists $\Gamma$ and "relations" $\+S_1 \in \+V_{\Gamma_1}, \hdots \+S_k \in \+V_{\Gamma_k}$
	such that $\SyntSA{\+R}$ "divides@@sync"
	$\?B \defeq \SyntSA{\+S_1} \times \cdots \times \SyntSA{\+S_k}$,
	"ie" there is a "subalgebra@@sync" $\?C \subseteq \?B$ which is a "quotient@@sync" of $\?B$.
	Then $\?C$ also "recognizes@@sync" $\+R$, say $\+R = \phi^{-1}[\Acc]$ for some
	"morphism@@sync" $\phi\colon \Sync\Sigma \surj \?C$ and $\Acc \subseteq \?C$.
	Let $\iota\colon \?C \to \?B$ be the canonical embedding,
	$\pi_i\colon \?B \surj \SyntSA{\+S_i}$ be the canonical projection,
	and $\phi_i \defeq \pi_i \circ \iota \circ \phi\colon \Sync\Sigma \to \SyntSA{\+S_i}$
	for $i \in \lBrack 1, k\rBrack$. Then notice that since $\SyntSAM{\+S_i}\colon \Sync{\Gamma_i} \surj \SyntSA{\+S_i}$ is surjective, then there exists $\psi_i \colon \Sync\Sigma \to \Sync{\Gamma_i}$ such that $\SyntSAM{\+S_i} \circ \psi_i = \phi_i$. Indeed, it suffices
	to send $\pair{a}{b}$ (resp.~$\pair{a}{\pad}$, resp.~$\pair{\pad}{a}$)
	on any element $\type{u}{\ll} \in \Sync{\Gamma_i}$ (resp.~$\type{u}{\lb}$,
	resp.~$\type{u}{\bl}$) such that $\SyntSAM{\+S_i}(\type{u}{\ll}) = \phi\pair{a}{b}$
	(resp.~$\SyntSAM{\+S_i}(\type{u}{\lb}) = \phi\pair{a}{\pad}$,
	resp.~$\SyntSAM{\+S_i}(\type{u}{\bl}) = \phi\pair{\pad}{a}$). Overall, the following diagram
	commutes
	\begin{center}
	\begin{tikzcd}
		\Sync\Sigma
			\ar[rr, "\psi_i"]
			\dar["\phi" swap]
			\ar[ddrr, "\phi_i"]
		&[-2em]
		& \Sync{\Gamma_i}
			\ar[dd, "\SyntSAM{\+S_i}"] \\
		\?C
			\ar[dr, "\iota" swap]
		&
		& \\[-1.5em]
		& \?B
			\ar[r, "\pi_i" swap]
		& \SyntSA{\+S_i}.
	\end{tikzcd}
	\end{center}
	Our goal is to show that $\+R \in \+V_{\Sigma}$. Observe that:
	\[
		\proj{\+R} =
		\phi^{-1}[\Acc] = \cup_{x \in \Acc} \phi^{-1}[x]
	\]
	but then $\Acc \subseteq \?B$, so $x$ is a tuple $\langle x_1, \hdots, x_n \rangle$
	(all elements having the same type), and by definition:
	\[
		\phi^{-1}[x] = \bigcap_{i=1}^n \phi^{-1}[\iota^{-1}[\pi_i^{-1}[x_i]]]
		= \bigcap_{i=1}^n \phi_i^{-1}[x_i].
	\]
	But then $\phi_i^{-1}[x_i] = \psi_i^{-1}[\SyntSAM{\+S_i}^{-1}[x_i]]$.
	Since $\+V$ is closed under "syntactic derivatives" and $\+S_i \in \+V_{\Gamma_i}$
	we have $\SyntSAM{\+S_i}^{-1}[x_i] \in \+V_{\Gamma_i}$, and then since $\+V$ is closed under
	"Inverse morphisms" and $\psi_i\colon \Sync\Sigma \to \Sync{\Gamma_i}$ is a "morphism@@sync" between free algebras,
	$\psi_i^{-1}[\SyntSAM{\+S_i}^{-1}[x_i]] \in \+V_{\Sigma}$.
	Thus $\proj{\+R}$ is a "Boolean combination" of elements of $\+V_{\Sigma}$, and hence
	it also belongs to $\+V_{\Sigma}$. This concludes the proof of $\+W \subseteq \+V$.
\end{proof}
\section{Synchronous Algebras Require a Dependency Relation}
\label{apdx:counterexample}

In this section, we introduce the notion of "synchronous algebra" with no dependency relation, 
called "naive synchronous algebra". This notion is more natural---or naive---than
\Cref{def:synchronous-algebra}, and share some of its enjoyable properties, such as the existence of syntactic algebras. Yet, we show that these algebras cannot characterize some
natural classes of synchronous relation. More precisely, we show that there
is a "$\ast$-pseudovariety of regular languages" $\+V$ and two synchronous relations
$\+R_0$ and $\+R_1$, such that:
\begin{itemize}
	\item $\+R_0$ is a $\+V$-relation,
	\item $\+R_1$ is not a $\+V$-relation,
	\item $\+R_0$ and $\+R_1$ have the same syntactic "naive synchronous algebra".
\end{itemize}

\begin{definition}[Naive synchronous algebra]
	Let $\intro*\typesEps \defeq \types \cup \{1\}$. We extend the notion of "compatibility" so 
	that every $\sigma \in \typesEps$ is "compatible" with $1$ and
	$1$ is "compatible" with $\sigma$.
	A \AP""naive synchronous algebra"" $\?A$ consists of a "$\typesEps$-typed sets",
	together with a partial binary operator $\cdot$ such that:
	\begin{itemize}
		\item $\cdot$ is defined exactly on "compatible" elements and is associative, and
		\item there is a unique element of type $1$, denoted by $1$, and it satisfies
			$\type{x}{\sigma} \cdot 1 = \type{x}{\sigma} = 1 \cdot \type{x}{\sigma}$
			for all $\type{x}{\sigma} \in \?A$.
	\end{itemize}
\end{definition}

The set of all "synchronous words" is naturally a "naive synchronous algebra" under concatenation.
Moreover, any "synchronous relation" admits a syntactic "naive synchronous algebra"---this can be shown in the same fashion as \Cref{lem:syntactic-morphism-theorem}.

\begin{example}[Group relations: \Cref{ex:charac-group-relation-monoids} cont'd]
	Consider the relations
	\begin{align*}
		\+R_0 \defeq \big\{
			(u,v) \;\big\vert\; & |u| > |v| \text{ and } (|u| - |v| \bmod{p}) = 0
		\big\} \\
		\+R_1 \defeq \big\{
			(u,v) \;\big\vert\; & |u| > |v| \text{ and } (|u| - |v| \bmod{p}) = 1
		\big\}.
	\end{align*}
	Then by \Cref{ex:charac-group-relation-monoids}, $\+R_0$ is not a "group relation" but
	$\+R_1$ is.
	Yet, we claim that both relations have the same syntactic "naive synchronous algebra" $\?A$, described as follows:
	\begin{itemize}
		\item it has a unit, denoted by $0$, of type $1$,
		\item $\?A_{\ll}$, $\?A_{\bl}$ and $\?A_{\ll\to\bl}$ are all reduced to a single element,
			denoted by $\type{0}{\ll}$, $\type{0}{\bl}$ and $\type{0}{\ll\to\bl}$,
		\item $\?A_{\ll\to\lb}$ and $\?A_{\lb}$ contain the elements $\Z/p\Z$,
		\item $\cdot$ is defined as the addition over $\Z/p\Z$, by identifying
			$\type{0}{1}$, $\type{0}{\ll}$, $\type{0}{\ll\to\bl}$ and $\type{0}{\bl}$
			with the zero of $\Z/p\Z$.
	\end{itemize}
	Then $\+R_0$ and $\+R_1$ are the preimages of $\{\type{0}{\ll\to\lb}, \type{0}{\lb}\}$
	and $\{\type{1}{\ll\to\lb}, \type{1}{\lb}\}$, respectively, by the natural morphism onto $\?A$.
	And hence $\+R_0$ and $\+R_1$ are recognized by $\?A$. It is easy to show that it is in fact the syntactic "naive synchronous algebra" of these relations: by surjectivity of the morphism 
	above, it suffices to show that no two elements of $\?A$ can be identified and still recognize 
	the same relation.
\end{example}

And so, from this example is follows that ``being a $\+V$-relations'' cannot be characterized
by the syntactic "naive synchronous algebra" of the relation, which shows how crucial the "dependency relation" of \Cref{def:synchronous-algebra} is in order to get \Cref{thm:lifting-theorem-monoids}.

The same result can be used to prove that ``naive positive synchronous algebras''---defined analogously to "naive synchronous algebra" except that there is no "type" 1 and no unit, and hence no empty word in the free algebra---are also unable to capture the property of ``being a $\+V$-relations''.

\clearpage
\bibliography{refs-algebras-automatic-relations}
\vfill
\setcounter{tocdepth}{1}
\tableofcontents

\end{document}